\documentclass[10pt, a4paper]{article}
\usepackage{latexsym,bm}
\usepackage{amssymb,amsmath,amsthm}
\usepackage{indentfirst}
\usepackage{graphicx}
\usepackage[T1]{fontenc}
\usepackage[utf8]{inputenc}
\usepackage{authblk}
\usepackage{subfigure}
\title{\textbf{Trajectory tracking control for nonholonomic mobile robots under arbitrary reference input}}

\author{Xiaodong He}
\author{Zhiyong Geng\thanks{Corresponding author: zygeng@pku.edu.cn}}
\affil{State Key Laboratory for Turbulence and Complex Systems, Department of Mechanics and Engineering Science, College of Engineering, Peking University, Beijing 100871, China}

\newcommand{\trans}{\begin{scriptsize}\mbox{T}\end{scriptsize}}

\newtheorem{theorem}{Theorem}
\newtheorem{lemma}{Lemma}

\theoremstyle{definition}
\newtheorem{definition}{Definition}
\newtheorem{remark}{Remark}
\newtheorem{example}{Example}

\date{}

\begin{document}

\graphicspath{{figures/}}

\maketitle

\noindent \textbf{Abstract:} This paper studies the tracking control problem for nonholonomic mobile robots based on second order dynamics, with application to consensus tracking and formation tracking. The greatest novelty in this paper is that the reference trajectory can be arbitrarily chosen, in the sense that the condition of persistency excitation and any other requirements are not imposed on the leader. In this paper, at first, the tracking control of one leader with single follower is taken into consideration, which is converted to the stabilization of two relative subsystems by the design of an adjoint system. Then, the single follower tracking controller is extended to the consensus tracking of multiple nonholonomic mobile robots connected by a directed acyclic graph, in which the convex combination in nonlinear manifolds is introduced to construct a virtual leader for each follower. Next, the relationship between consensus control and formation control is established for multiple nonholonomic mobile robots, with the result that a new transformed system is constructed so as to obtain the formation tracking control strategy from the consensus tracking result. Finally, simulations are presented to verify the effectiveness of the control laws.

\


\section{Introduction}

Cooperative control for networked multi-agent systems has received considerable attention during the past decade, to a certain extent motivated by its widely applications in various areas \cite{Fax,Oh,SunZY}. Particularly, numerous researchers have shown great interest in the trajectory tracking control for nonholonomic mobile robots, such as \cite{Jiang,Dixon,DoJ,Rudra,SunZ} and references therein. In the tracking control problem, the case where multiple followers are forced to achieve common values with the leader is generally mentioned as consensus tracking control. If the followers are required to move towards and maintain a desired geometric pattern with respect to the leader, then the problem can be referred to as formation tracking control.

The problem of nonholonomic mobile robots formation tracking or leader-follower formation has been investigated in a large number of literatures. In \cite{LiuJ}, a distributed controller for leader-follower formation was presented by the nonlinear small gain method. The paper \cite{Vos} investigated the formation problem under the port-Hamiltonian framework and presented a controller with the assumption that initial attitude angles of robots lie in a semicircle. In \cite{Loria}, the leader-follower formation controller was designed with the reference trajectory as a straight line path. The authors in \cite{Maghenem} proposed a nonlinear time-varying controller under the relaxed persistency of excitation condition of velocities' norm. In \cite{Cheng}, a disturbance observer-based output feedback controller was developed by the method of integral sliding-mode control.

Additionally, several works handle the formation tracking problem in a more complex case or with other constraints, which make the controller more practical and applicable in real situation. Author in \cite{Do} considered limited sensing range of mobile robots and collision avoidance when designing the formation tracking controller. In \cite{Yu}, the velocity constraints of mobile robots were taken into account in the formation controller design. To obtain the leader's states, in \cite{Miao}, a distributed estimator is proposed for each follower and the formation controller is designed based on the estimated information. In \cite{Chen}, the authors presented a receding-horizon controller to sharpen the performance of formation tracking errors convergence.

Above literatures are all relevant and significant works in the tracking control problem, and remarkably enrich and the nonlinear control theory for nonholonomic mobile robot, while there still exist several aspects that could be further improved, which are summarized in the following.

Firstly, most of existing results in leader-follower formation requires the persistency of excitation (PE) condition, roughly speaking, which means the velocity of the leader cannot be zero, or some equivalent conditions. In such cases, extra limitations are imposed on the trajectory of the leader, and the tracking controller is confined to a certain scope of application. However, in common reality, the mobile robots are able to freely move and stop in their motion, and the tracking controller should be effective under any situation. Thus, it would be more practical if the PE condition is removed and the trajectory of the leader can be arbitrarily chosen.

Secondly, a vast number of researches consider the states of the robots in linear space, while the actual configuration space is a nonlinear manifold. Only in exceptional circumstances can the configuration be described by vectors in the Euclidean space \cite{Bullo2005}. The most significant advantage of the manifold description lies in globalness and uniqueness. Specifically, the motion of mechanical systems can be described in a manifold independent of the local coordinates, contributing to global derived results. Therefore, in this paper the robot is modelled in the Lie group SE(2), with the full name as Special Euclidean group in two dimensions, which is an important class of manifolds for planar rigid body motion. The researches under Lie group framework can be found in \cite{Dong,LiuND,LiuRNC,Peng,Sun,TayefiND}.

Thirdly, plenty of results are mainly based on the kinematic model, in which the control input is velocity and steps and impulses are absolutely prohibited due to the physical nature. However, owing to no such requirements for force and torque, the controller at dynamic level are more flexible and practical. Some control strategies on dynamics are derived from kinematics controls by the backstepping method, which still requests the control algorithm to be designed on kinematics firstly. More importantly, in 3-dimensional space, it is extremely difficult to use the backstepping method to obtain control inputs on dynamics from those on kinematics, due to the gyroscopic torques. It would be more decent if the controller can be designed on dynamics directly.

Motivated by the existing literatures, in this paper we consider the the trajectory tracking control problem for nonholonomic mobile robots based on second order dynamics, and extend the controller to consensus tracking and formation tracking. The contributions of this paper are threefold.

\hangafter 1
\hangindent 2em
$\bullet$ At first, the tracking control problem of one leader with single follower is taken into consideration. By the design of an adjoint system, the original dynamics is decomposed into two subsystems. Thus, the tracking problem is able to be solved by designing the controller for each subsystem, which can be obtained based on several existing results.

\hangafter 1
\hangindent 2em
$\bullet$ Then, the single follower tracking controller is generalized to the consensus tracking of multiple nonholonomic mobile robots, which are connected by a directed acyclic graph with one root node. We introduce the convex combination in nonlinear manifolds to construct a local virtual leader for each follower, with the result that the consensus tracking is able to be converted to the single follower tracking.

\hangafter 1
\hangindent 2em
$\bullet$ Lastly, by coordinate transformation between different body-fixed frames, the relationship between consensus control and formation control is established for multiple nonholonomic mobile robots. Based on such a relationship, a new transformed system is proposed so as to obtain the formation tracking control strategy from the consensus tracking result.

The novelties of the presented tracking control strategy are as follows. 1) The reference trajectory of the leader can be arbitrarily chosen, that is to say, the PE condition and any other requirements are all removed from the leader. Thus, the tracking controller of the follower can be effective under all types of reference input. 2) The convergence domain of the controller is global, in other words, the robots are able to accomplish the tracking task from any initial position and orientation. Furthermore, the initial relative attitude angle specified as $\pi$ or $-\pi$ would correspond to different trajectories, which provides choices in case of limitations imposed by surroundings. 3) The control strategy is directly proposed based on the dynamic model, rather than be designed on kinematics and derived to dynamics by the backstepping method. Thus, such a controller in dynamics is possible to be extended to 3-dimensional motion.

With regards to the control of nonholonomic mobile robot at the dynamics level, the tracking problem under global convergence and arbitrary reference input possessed general universality and practicality. Unfortunately, the control strategy for such a kind of problem cannot be found in the existing literatures. To the best of our knowledge, it is the first time that a global tracking controller is proposed for arbitrary reference input in the sense of dynamics.

The remainder of this paper is organized as follows. The mathematical preliminaries and problem formulation are given in Section 2. In Section 3, we present the controller for single follower tracking, consensus tracking and formation tracking respectively. Numerical simulation examples are presented in Section 4 to verify the control algorithm. Finally, the paper comes to an end with the conclusion in Section 5.

\section{Preliminaries and Problem Formulation}

The mathematical preliminaries and problem formulation are provided in this section, where the notations used in this paper are introduced in Table \ref{tab_nontation}.

\begin{table}\small
\caption{Notations used in this paper}
\label{tab_nontation}
\centering
\begin{tabular}{ll}
\hline\noalign{\smallskip}
Symbol  &  Description \\
\noalign{\smallskip}\hline\noalign{\smallskip}
    $x$, $y$   & Position of the rigid body in the plane  \\
    $\bm{p}$   & Position vector, i.e. $\bm{p}=[x \ \ y]^{\begin{tiny}\mbox{T}\end{tiny}}$  \\
    $\theta$   & Attitude angle of the rigid body in the plane   \\
    $R$        & Rotation matrix through $\theta$   \\
    $v_x$, $v_y$ & Translational velocity of the rigid body   \\
    $\bm{v}$   & Translational velocity vector, i.e. $\bm{v}=[v_x \ \ v_y]^{\begin{tiny}\mbox{T}\end{tiny}}$  \\
    $\omega$   & Angular velocity of the rigid body  \\
    $g$        & Lie group element representing configuration \\
    $\hat{\xi}$  & Lie algebra element representing velocity, with $\bm{\xi}$ as its vector form  \\
    $\hat{X}$  & Exponential coordinates of $g$, with $\bm{X}$ as its vector form \\
    $\hat{u}$  & Control input at force and torque level, with $\bm{u}$ as its vector form \\
    $\bm{q}$   & Translational component of $\hat{X}$ and $\bm{q}=[q_x\ \ q_y]^{\begin{tiny}\mbox{T}\end{tiny}}$  \\
    $\alpha$   & Auxiliary angular term and $\alpha=(\theta/2)\cot(\theta/2)$  \\
    $\beta$    & Angle of the vector $\bm{q}$ in the body-fixed frame, i.e. $\beta=-\arctan(q_y/q_x)$  \\
    $g_{ij}$   & Relative configuration of $g_j$ with respect to $g_i$ \\
    $\hat{\xi}_{ij}$ & Relative velocity of $\hat{\xi}_j$ with respect to $\hat{\xi}_i$, with $\bm{\xi}_{ij}$ as its vector form  \\
    $\bar{g}_{ij}$   & Desired relative configuration of $g_j$ with respect to $g_i$ \\
    $g_{c_i}$  & Convex combination of the parent nodes of $g_{i}$ \\
    $\bar{g}_{c_{i}i}$  & Desired relative configuration of $g_i$ with respect to $g_{c_i}$ \\
    $g_{a_i}$  & Auxiliary transformed system of $g_{i}$ and $g_{a_{i}}=g_{c_{i}}\bar{g}_{c_{i}i}$ \\
    $I$        & Identity matrix \\
\hline\noalign{\smallskip}
\end{tabular}
\end{table}

\subsection{Systems described in the Lie group SE(2)}

Consider a rigid body moving in a plane. Let $\mathcal{F}_s$ denote the spacial or inertial frame attached to the earth, and let $\mathcal{F}_b$ represent the body-fixed frame, which is attached to the center of mass of the rigid body. The position of the rigid body in $\mathcal{F}_s$ is described by a vector $\bm{p}=[x \ \ y]^T\in\mathbb{R}^2$. The orientation of the rigid body is specified by a rotation matrix $R\in\mbox{SO(2)}$, where the Special Orthogonal group is defined as $\mbox{SO(2)}=\{R\in\mathbb{R}^{2\times 2}\ |\ R^TR=I_2,\mbox{det}R=1\}$. Specifically, $R$ is expressed as $R=\begin{bmatrix} \cos\theta & -\sin\theta \\ \sin\theta & \cos\theta \end{bmatrix}$, where $\theta$ is the attitude angle of $\mathcal{F}_b$ relative to $\mathcal{F}_s$. Thus, the configuration of the rigid body is able to be denoted by a matrix \begin{eqnarray} \label{eq_defi_g} g=\begin{bmatrix} R & \bm{p} \\ \textbf{0}_{1\times 2} & 1 \end{bmatrix}=\begin{bmatrix} \cos\theta & -\sin\theta & x \\ \sin\theta & \cos\theta & y \\ 0 & 0 & 1 \end{bmatrix}\in\mathbb{R}^{3\times 3}.\end{eqnarray} The set of all such configurations $g$, with matrix multiplication as a group operation on it, constitutes a matrix Lie group named Special Euclidean group SE(2). In other words, $g$ is the element in SE(2). Hence, the Lie group SE(2) is the configuration space of a planar rigid body, capturing the position and orientation simultaneously.

Denote $T_g\mbox{SE(2)}$ the tangent space of SE(2) at element $g$, so that $T_I\mbox{SE(2)}$ is the tangent space at identity element $I$. Define the following skew symmetric bilinear operation on $T_I\mbox{SE(2)}$ \begin{eqnarray*}[\hat{X}, \hat{Y}]=\hat{X}\hat{Y}-\hat{Y}\hat{X}, \quad \hat{X},\hat{Y}\in T_I\mbox{SE(2)},\end{eqnarray*} which is called Lie bracket. For a given $\hat{X}$, the Lie bracket $[\hat{X}, \hat{Y}]$ defines a linear map on $T_I\mbox{SE(2)}$ denoted by $\mbox{ad}_{\hat{X}}$, which is called the adjoint map defined by $\hat{X}$. Thus, there holds that $\mbox{ad}_{\hat{X}}\hat{Y}=[\hat{X}, \hat{Y}]$. Once the linear space $T_I\mbox{SE(2)}$ is endowed with Lie bracket, it is named as Lie algebra and denoted by $\mathfrak{se}(2)$. The element $\hat{\xi}\in\mathfrak{se}(2)$ is defined as \begin{eqnarray} \label{eq_defi_xi} \hat{\xi}=\begin{bmatrix} \hat{\omega} & \bm{v} \\ \bm{0}_{1\times 2} & 0 \end{bmatrix} = \begin{bmatrix} 0 & -\omega & v_x \\ \omega & 0 & v_y \\ 0 & 0 & 0 \end{bmatrix} \in \mathbb{R}^{3\times 3}, \end{eqnarray} where $\hat{\omega}\in\mathfrak{so}(2)$ (the Lie algebra associated with SO(2)) is a skew symmetric matrix representing the angular velocity, and $\bm{v}=[v_{x}\ \ v_{y}]^{\trans}\in\mathbb{R}^2$ is a vector denoting the translational velocity. Herein, $\wedge:\mathbb{R}\to\mathfrak{so}(2)$ is called hat map while its inverse is vee map $\vee:\mathfrak{so}(2)\to\mathbb{R}$. For a 3-dimensional vector $\bm{\xi}=[\omega \ \ v_x \ \ v_y]^T$, the $\hat{\xi}$ is defined in (\ref{eq_defi_xi}), with slight abuse of notation. As is portrayed in Figure \ref{fig_sys}, $g$ describes the position and orientation, and $\hat{\xi}$ represents the velocity, where $v_x$ and $v_y$ are along the positive direction of the body-fixed frame $\mathcal{F}_b$.

For a given $g\in\mbox{SE(2)}$, it can be shown that $g\dot{g}$ and $\dot{g}g$ are both in $\mathfrak{se}(2)$. The former is called the body velocity denoted by $\hat{\xi}$ (actually it is the velocity that we define in (\ref{eq_defi_xi})), which is expressed in the body-fixed frame $\mathcal{F}_b$, while the latter is called the spacial velocity denoted by $\hat{\xi}^s$, which is expressed in the spacial frame $\mathcal{F}_s$. The relation of these two velocities is given by $\hat{\xi}^s=\mbox{Ad}_{g}\hat{\xi}$, in which $\mbox{Ad}_{g}:\mathfrak{se}(2)\to\mathfrak{se}(2)$ is named the adjoint map and defined as \begin{equation} \label{eq_Ad} \mbox{Ad}_{g}\hat{\xi}=g\hat{\xi}g^{-1}. \end{equation} Hence, such a map plays an important role in the reference frame transformation for the velocity.

\begin{figure}
\centering
\includegraphics[width=2.2in]{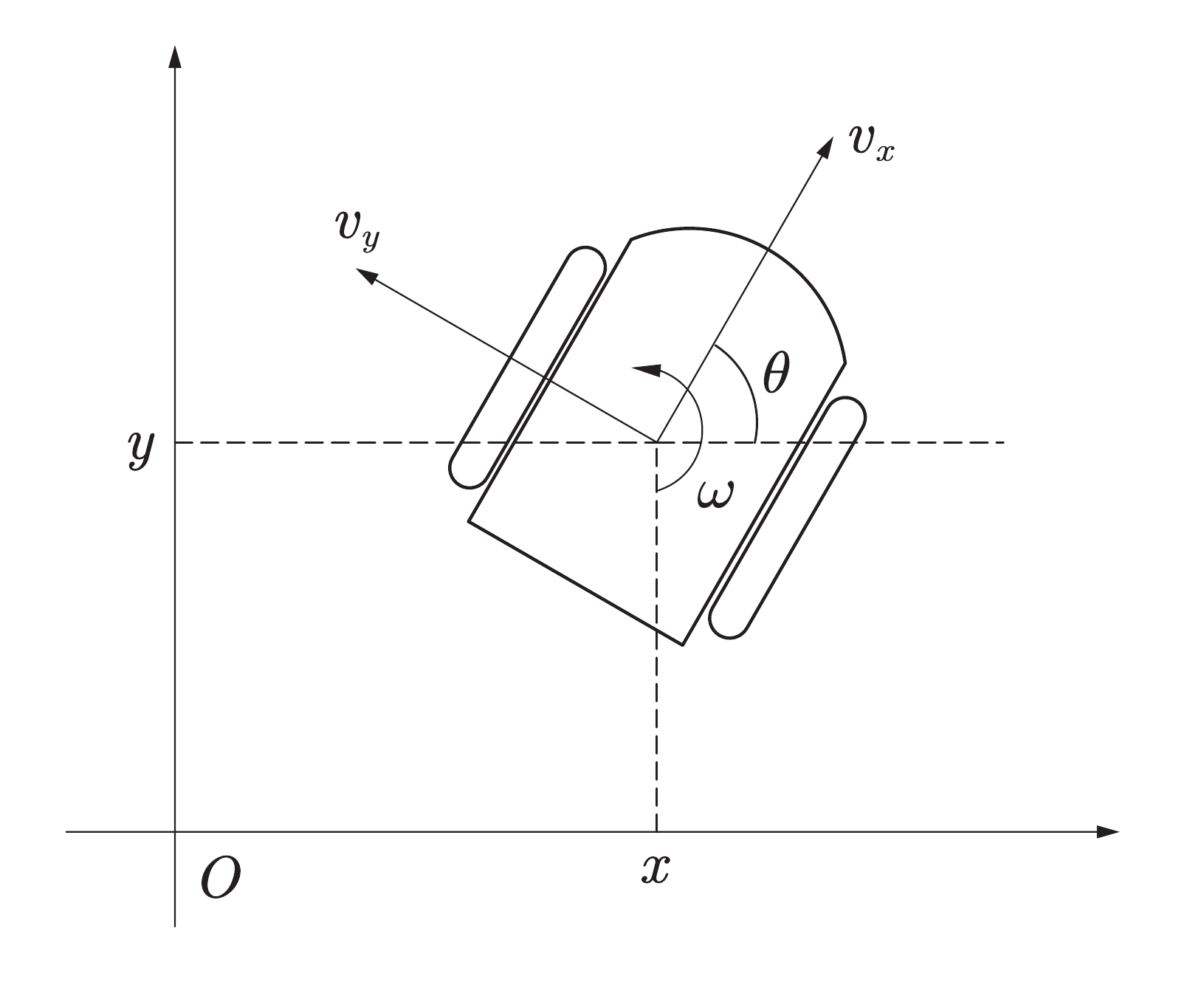}
\caption{Coordinates of a planar rigid body}
\label{fig_sys}
\end{figure}

\begin{definition}
For a planar rigid body, it is said to be nonholonomic constrained if the time derivative of the position vector $\bm{p}$ satisfies \begin{equation*} \label{non} [-\sin\theta \quad \cos\theta]\dot{\bm{p}}=0.\end{equation*} It can be shown that this is equivalent to $v_y=0$, which is also called the nonholonomic constraint.
\end{definition}

The model that would be considered in this paper is the nonholonomic mobile robot.

\subsection{Dynamics of the system}

The dynamics of the rigid body evolving on the tangent bundle $T\mbox{SE(2)}$ can be described by the $Euler-Poincar\acute{e}$ equation \begin{equation*} \left\{ \begin{aligned} & \dot{g}=g\hat{\xi} \\ & \dot{\hat{\xi}} = [\mathbb{I}]^{-1} ([\mbox{ad}_{\hat{\xi}}]^{\trans} [\mathbb{I}] \hat{\xi} + F) \end{aligned} \right. , \end{equation*} in which $[\mathbb{I}]$ is the inertial matrix, $[\mbox{ad}_{\hat{\xi}}]$ is the adjoint operator, and $F$ is the control force or torque. Generally, $[\mathbb{I}]$ and $[\mbox{ad}_{\hat{\xi}}]$ can be explicitly expressed in SE(2) as \begin{eqnarray*}[\mathbb{I}]=\begin{bmatrix} J & & \\ & m &  \\ & & m \end{bmatrix}, \quad [\mbox{ad}_{\hat{\xi}}]=\begin{bmatrix} 0 & 0 & 0 \\ v_y & 0 & -\omega \\ -v_x & \omega & 0 \end{bmatrix},\end{eqnarray*} where $J$ is the moment of inertia and $m$ is the mass.

For fully actuated systems, we can design $F$ compensating the drift term so as to express the equation in an integrator form. In other words, let $\hat{u}=[\mathbb{I}]^{-1} ([\mbox{ad}_{\hat{\xi}}]^T [\mathbb{I}] \hat{\xi} + F)$, then the dynamics of the system can be written as \begin{equation} \label{sys} \left\{ \begin{aligned} & \dot{g}=g\hat{\xi} \\  & \dot{\hat{\xi}}=\hat{u} \end{aligned}, \right. \end{equation} where $g$ and $\hat{\xi}$ are the states of the system, and $\hat{u}$ is the control input at the force and torque level. The vector form of the control input can be denoted by $\bm{u}=[u_{\theta}\ \ u_{x}\ \ u_{y}]^{\trans}$.  Hence, once having designed control input $\hat{u}$, we are able to obtain the real control forces and torques by the relationship $F=[\mathbb{I}]\hat{u}-[\mbox{ad}_{\hat{\xi}}]^T [\mathbb{I}] \hat{\xi}$.

When the nonholonomic system is taken into account, due to no sideways slip, it is actually a kind of underactuated system. However, as a matter of fact, only the input channel of $\dot{v}_y$ is constrained, while $\dot{\omega}$ and $\dot{v}_x$ are still no constraint channels. Thus, in the direction of $\dot{\omega}$ and $\dot{v}_x$, the above compensation can still be used to reach an integrator form.

In fact, in the direction of $\dot{v}_y$, the dynamic equation can also be written in a similar form. Owing to the nonholonomic constraint, on one hand, the side velocity $v_y=0$, resulting in that the time derivative $\dot{v}_y=0$. On the other hand, no control force can be exerted along the lateral direction, i.e. $u_y=0$. If any, it would be balanced by the constraint force. Thus, $\dot{v}_y$ and $u_y$ are both 0, leading to the equality $\dot{v}_y=u_y$, which is in an integrator form as well. Therefore, for the nonholonomic system, the dynamic model can also be written as (\ref{sys}), where the control input is $\bm{u}=[u_{\theta}\ \ u_{x}\ \ 0]^{\trans}$.

\subsection{Exponential Coordinates}

In the following part, we introduce the exponential coordinates and their time derivatives, which are acquired from \cite{Bullo1995}. Exponential map on matrices is defined as \begin{equation*}\mbox{exp}A=\sum^{+\infty}_{k=0}\frac{A^k}{k!},\end{equation*} which is $\mathbb{R}^{n\times n}\to\mathbb{R}^{n\times n}$. In the Lie group SE(2), it has a more explicit form. Given $\hat\theta\in\mathfrak{so}(2)$, $\bm{q}\in\mathbb{R}^2$ and $\hat{X}=(\hat\theta, \bm{q})\in\mathfrak{se}(2)$, the exponential map $\mbox{exp}_{\begin{scriptsize}\mbox{SE(2)}\end{scriptsize}}:\mathfrak{se}(2)\to\mbox{SE(2)}$ becomes \begin{equation*} g=\mbox{exp}_{\begin{scriptsize}\mbox{SE(2)}\end{scriptsize}}(\hat{X})= \begin{bmatrix} \mbox{exp}_{\begin{scriptsize}\mbox{SO(2)}\end{scriptsize}} (\hat{\theta}) & A(\theta)\bm{q} \\ \bm{0} & 1  \end{bmatrix}, \end{equation*} in which there holds $\mbox{exp}_{\begin{scriptsize}\mbox{SO(2)}\end{scriptsize}} (\hat{\theta})=\begin{bmatrix} \cos\theta & -\sin\theta \\ \sin\theta & \cos\theta \end{bmatrix}$ and $A(\theta)=\frac{\displaystyle 1}{\displaystyle \theta}\begin{bmatrix} \sin\theta & -(1-\cos\theta) \\ (1-\cos\theta) & \sin\theta \end{bmatrix}$.

The inverse of the exponential map is named logarithm map $\mbox{log}_{\begin{scriptsize}\mbox{SE(2)}\end{scriptsize}}:\mbox{SE(2)}\to\mathfrak{se}(2)$. Provided $R\in\mbox{SO(2)}$, $\bm{p}\in\mathbb{R}^2$, $g=(R,\bm{p})\in\mbox{SE(2)}$ and $\mbox{trace}(g)\ne-1$, $\mbox{log} _{\begin{scriptsize}\mbox{SE(2)}\end{scriptsize}}$ is defined as \begin{equation} \label{ex_cord} \hat{X}=\mbox{log}_{\begin{scriptsize}\mbox{SE(2)}\end{scriptsize} }(g)=\begin{bmatrix} \hat\theta & A^{-1}(\theta)\bm{p} \\ \bm{0} & 0 \end{bmatrix}, \end{equation} where $\hat\theta=\mbox{log} _{\begin{scriptsize}\mbox{SO(2)}\end{scriptsize}}(R)$, $A^{-1}(\theta)=\begin{bmatrix} \alpha(\theta) & \theta/2 \\ -\theta/2 & \alpha(\theta) \end{bmatrix}$ and $\alpha(\theta)=(\theta/2)\cot(\theta/2)$. Denote that $\bm{q}=[q_x \ \ q_y]^T=A^{-1}(\theta)\bm{p}$, then $\hat{X}=\mbox{log}_{\begin{scriptsize}\mbox{SE(2)}\end{scriptsize}} (g)=(\hat\theta,\bm{q})$ are referred to as the exponential coordinates of group element $g$, whose vector form is \begin{equation*} \bm{X}=\begin{bmatrix} \theta & q_{x} &q_{y} \end{bmatrix}^{\trans}. \end{equation*} Note that for the case of $\mbox{trace}(g)=-1$, $\hat{\theta}=\mbox{log} _{\begin{scriptsize}\mbox{SO(2)}\end{scriptsize}}(R)$ is a two valued function corresponding to $-\pi$ or $\pi$. Thus, we can specify its value as needed, which leads to more choices when designing the control law.

In the following, we introduce a stabilization control law using exponential coordinates, which will be employed in the controller design later.

\begin{lemma}[\cite{He}, Theorem 1]
\label{le_sta_nonho}
For a nonholonomic mobile robot evolving in SE(2) with dynamics (\ref{sys}), the system state $g$ can be stabilized to the identity element from arbitrary initial configuration under the following control law \begin{equation*} u=-k_p\begin{bmatrix} \theta+k\beta \\ q_x \\ 0 \end{bmatrix} -k_d\begin{bmatrix} \omega \\ v_x \\ 0 \end{bmatrix}, \end{equation*} where $\beta=-\arctan(q_y/q_x)$, $\bm{q}=[q_x \ \ q_y ]^{\begin{scriptsize}\mbox{T}\end{scriptsize}} =A^{-1}(\theta)\bm{p}$, and $k_p$, $k_d$, $k$ are all scalar control gains.
\end{lemma}

\subsection{Problem Formulation}

Consider $N+1$ nonholonomic mobile robots evolving in SE(2) labeled as $0,1,\cdots,N$, where robot $0$ is the leader and others are the followers. The dynamics of robot $i$ is described by the following equations \begin{equation} \label{eq_sysi} \Sigma_{i}:\left\{\begin{aligned} & \dot{g}_i=g_i\hat{\xi}_i \\ & \dot{\hat{\xi}}_i = \hat{u}_i \end{aligned},\right. i=0,1,\cdots,N, \end{equation} where $g_i\in\mbox{SE(2)}$ is the configuration of the $i$th robot, $\hat{\xi}_i\in\mathfrak{se}(2)$ is the velocity and $\hat{u}_i$ represents its control input. It is assumed that the control input of the leader has been predefined and independent of other followers. In the following, we define the relative configuration $g_{0i}$ in SE(2).

Referring to the definition of error in Lie groups, the relative configuration of robot $i$ with respect to the leader robot $0$ can be defined as \begin{equation} g_{0i}=g^{-1}_0g_i,\quad i=0,1,\cdots,N. \end{equation} Taking the derivative of $g_{0i}$ with respect to time $t$, we have \begin{align}\dot{g}_{0i} &= \frac{\mbox{d}}{\mbox{d}t}(g_{0}^{-1})g_i + g_{0}^{-1}\dot{g}_{i} \nonumber \\ &= -g_{0}^{-1}\dot{g}_{0}g_{0}^{-1}g_{i} + g_{0}^{-1}\dot{g}_{i}, \label{eq_dot_g01_m}\end{align} Substitute (\ref{eq_sysi}) into (\ref{eq_dot_g01_m}), it is obtained \begin{align} \dot{g}_{0i} &= -g_{0}^{-1}g_{0}\hat{\xi}_{0}g_{0}^{-1}g_{i} + g_{0}^{-1}g_{i}\hat{\xi}_{i} \nonumber \\ &= g_{0i}(\hat{\xi}_i-\mbox{Ad}_{g^{-1}_{0i}}\hat{\xi}_0), \label{eq_dot_g01_m2} \end{align} where the adjoint map in (\ref{eq_Ad}) is employed in the derivation and the term $\mbox{Ad}_{g^{-1}_{0i}}\hat{\xi}_0$ represents the leader's velocity $\hat{\xi}_{0}$ expressed in the body frame of follower $i$. Similarly, for any two robots $i$ and $j$, the relative configuration of robot $j$ with respect to robot $i$ is defined as $g_{ij}=g_{i}^{-1}g_{j}$, whose time derivative is calculated as $\dot{g}_{ij} = g_{ij}(\hat{\xi}_j-\mbox{Ad}_{g^{-1}_{ij}}\hat{\xi}_i),\ i,\ j=0,1,\cdots,N$.

It is said that nonholonomic mobile robots realize the desired formation if their states satisfy the following equations \begin{subequations} \label{eq_for_0i} \begin{align} &\lim_{t\to\infty}g_{0i}(t)=\bar{g}_{0i}, \\ &\lim_{t\to\infty}(\hat{\xi}_{i}(t)-\mbox{Ad}_{g_{0i}^{-1}(t)}\hat{\xi}_{0}(t))=0, \end{align} \end{subequations} or equivalently, \begin{subequations} \label{eq_for_ij} \begin{align} \label{eq_for_ij_g} &\lim_{t\to\infty}g_{ij}(t)=\bar{g}_{ij}, \\ \label{eq_for_ij_xi} &\lim_{t\to\infty}(\hat{\xi}_{j}(t)-\mbox{Ad}_{g_{ij}^{-1}(t)}\hat{\xi}_{i}(t))=0, \end{align} \end{subequations} where $\bar{g}_{0i}\in\mbox{SE}(2)$ is a constant group element provided in advance relying on the formation task. That is to say, $\bar{g}_{0i}$ is able to be interpreted as the desired relative configuration of robot $i$ with respect to leader robot $0$, which uniquely decides the geometric pattern of the formation. Similarly, $\bar{g}_{ij}$ represents the desired configuration of robot $j$ with respect to robot $i$. Once $\bar{g}_{0i}$ has been specified depending on the formation task, $\bar{g}_{ij}$ will be determined naturally by the relationship $\bar{g}_{ij}=\bar{g}_{0i}^{-1}\bar{g}_{0j}$. The desired formation $\bar{g}_{0i}$ contains the desired relative position vector $\bar{\bm{p}}_{01}$ and desired relative attitude angle $\bar{\theta}_{0i}$. It should be emphasized that $\bar{\bm{p}}_{01}$ and $\bar{\theta}_{0i}$ are not independent for nonholonomic mobile robot formation. The relation between them will be interpreted in the following Section.

The velocity requirement (\ref{eq_for_ij_xi}) implies that the velocity of robot $i$, which is observed in the body-fixed frame of robot $j$, will be asymptotically equal to the velocity of robot $j$. This is a necessary requirement for the formation maintenance. If the desired formation requirement (\ref{eq_for_ij_g}) is maintained, the space velocity, or the velocity expressed in the spacial frame $\mathcal{F}_s$, of the mobile robots should be equivalent, i.e. $\hat{\xi}_{i}^{s}=\hat{\xi}_{j}^{s}$. Although $\hat{\xi}_{i}$ and $\hat{\xi}_{j}$ are the body velocities of robot $i$ and robot $j$ in their own body frames, we can obtain the space velocities by the coordinate transformation $\mbox{Ad}_g$, that is, $\hat{\xi}_{i}^{s}=\mbox{Ad}_{g_{i}}\hat{\xi}_{i}$ and $\hat{\xi}_{j}^{s}=\mbox{Ad}_{g_{j}}\hat{\xi}_{j}$. Hence, from the relationship $\mbox{Ad}_{g_{i}}\hat{\xi}_{i}=\mbox{Ad}_{g_{j}}\hat{\xi}_{j}$, it is easy to acquire that $\hat{\xi}_{j}-\mbox{Ad}_{g_{ij}^{-1}}\hat{\xi}_{i}=0$, which is the velocity requirement (\ref{eq_for_ij_xi}).

Therefore, the problem to be solved in this paper is to design control input $\hat{u}_{i}$, specifically to design $u_{\theta}$ and $u_{x}$ for each nonholonomic mobile robot $i$ ($i=1,2,\cdots,N$), which is able to achieve the formation requirement (\ref{eq_for_0i}) or (\ref{eq_for_ij}). It should be noted that if the desired formation is chosen as $\bar{g}_{0i}=I$, it is able to be obtained $g_{ij}(t)\to I$ and $\hat{\xi}_{j}\to\hat{\xi}_{i}$ as $t\to\infty$, indicating that the formation problem degenerates to the consensus problem of all the robots.

\section{Control Strategy Design}

In this section, above all, we consider the tracking control problem under the case of single follower, then investigate the consensus tracking problem for multiple nonholonomic mobile robots, and finally present the control strategy for formation tracking problem.

\subsection{Single Follower Tracking}

The tracking control problem of one leader with single follower is studied at first. The configurations of leader and follower are denoted by $g_0$ and $g_1$ respectively, so that the relative configuration of the follower with respect to the leader is defined as $g_{01}=g_{0}^{-1}g_1$, that is, \begin{align} g_{01}=\begin{bmatrix} R^{\trans}_0 & -R^{\trans}_0\bm{p}_0 \\ \bm{0} & 1 \end{bmatrix}\begin{bmatrix} R_1 & \bm{p}_1 \\ \bm{0} & 1 \end{bmatrix} = \begin{bmatrix} R_{0}^{\trans}R_{1} & R_{0}^{\trans}(\bm{p}_{1}-\bm{p}_{0}) \\ \bm{0} & 1 \end{bmatrix}  \triangleq  \begin{bmatrix} R(\theta_{01}) & \bm{r}_{01} \\ \bm{0} & 1 \end{bmatrix}, \label{eq_g01} \end{align} where the rotation matrix is written as \begin{equation*} R(\theta_{01})=\begin{bmatrix} \cos\theta_{01} & -\sin\theta_{01} \\ \sin\theta_{01} & \cos\theta_{01} \end{bmatrix}, \end{equation*} and the position vector is $\bm{r}_{01}=[r_{x01}\ \ r_{y01}]^{\trans}$. The definitions of $\theta_{01}$, $r_{x01}$ and $r_{y01}$ are that $\theta_{01}=\theta_1-\theta_0$, $r_{x01}=(x_1-x_0)\cos{\theta_0}+(y_1-y_0)\sin{\theta_0}$, and $r_{y01}=-(x_1-x_0)\sin{\theta_0}+(y_1-y_0)\cos{\theta_0}$. Define the relative velocity as \begin{equation} \label{eq_xi_01_0} \hat{\xi}_{01}=\hat{\xi}_1-\mbox{Ad}_{g^{-1}_{01}}\hat{\xi}_0, \end{equation} then according to (\ref{eq_dot_g01_m2}), the kinematics of the relative system can be expressed as \begin{equation} \dot{g}_{01}=g_{01}\hat{\xi}_{01}. \end{equation} Before deriving the dynamics of the relative velocity $\hat{\xi}_{01}$, we introduce the following Lemma.

\begin{lemma}[\cite{Dong}, Lemma 6]
\label{le_dong}
For the relative configuration $g_{01}$ and relative velocity $\hat{\xi}_{01}$, there holds the following properties. \begin{align*}  & \frac{\mbox{d}}{\mbox{d}t}(\mbox{Ad}_{g_{10}})\hat{\xi}_{0}= \mbox{Ad}_{g_{10}}[\hat{\xi}_{10},\hat{\xi}_{0}], \\ & \mbox{Ad}_{g_{10}}[\hat{\xi}_{0},\hat{\xi}_{10}]= -[\hat{\xi}_{1},\hat{\xi}_{01}]. \end{align*}
\end{lemma}

Now, we compute the time derivative of the relative velocity $\hat{\xi}_{01}=\hat{\xi}_1-\mbox{Ad}_{g^{-1}_{01}}\hat{\xi}_0$, that is \begin{equation} \label{eq_dot_rexi_m} \dot{\hat{\xi}}_{01} = \dot{\hat{\xi}}_{1} - \mbox{Ad}_{g_{01}^{-1}}\dot{\hat{\xi}}_{0} - \frac{\mbox{d}}{\mbox{d}t}(\mbox{Ad}_{g_{01}^{-1}})\hat{\xi}_{0}. \end{equation} According to Lemma \ref{le_dong}, there holds \begin{equation} \label{eq_dot_ad} \frac{\mbox{d}}{\mbox{d}t}(\mbox{Ad}_{g_{01}^{-1}})\hat{\xi}_{0} = [\hat{\xi}_{1},\hat{\xi}_{01}]. \end{equation} Substitute (\ref{eq_sysi}) and (\ref{eq_dot_ad}) into (\ref{eq_dot_rexi_m}), then the time derivative of the relative velocity $\hat{\xi}_{01}$ is obtained as \begin{equation} \label{eq_dot_rexi} \dot{\hat{\xi}}_{01} = \hat{u}_{1} - \mbox{Ad}_{g_{01}^{-1}}\hat{u}_{0} - [\hat{\xi}_{1},\hat{\xi}_{01}]. \end{equation} Thus, the dynamics of the relative system is \begin{equation} \label{eq_sys01} \Sigma_{01}: \left\{ \begin{aligned}  & \dot{g}_{01}=g_{01}\hat{\xi}_{01} \\ & \dot{\hat{\xi}}_{01} = \hat{u}_{01} \end{aligned}, \right. \end{equation} where the control input is $\hat{u}_{01} = \hat{u}_{1} - \mbox{Ad}_{g_{01}^{-1}}\hat{u}_{0} - [\hat{\xi}_{1},\hat{\xi}_{01}]$.

Having derived the dynamics of the relative system $\Sigma_{01}$, we are able to convert the problem of tracking into that of control for $\Sigma_{01}$, which is interpreted in the following Lemma.

\begin{lemma}
\label{le_sys_01}
The follower achieves the trajectory tracking with respect to the leader, if the relative system $\Sigma_{01}$ can be stabilized to the identity element.
\end{lemma}

\begin{proof}
If the relative system $\Sigma_{01}$ is stabilized with the control input $\hat{u}_{01}$, there holds $g_{01}\to I$ and $\hat{\xi}_{01}\to 0$ as $t\to\infty$. According to the definition of $\hat{\xi}_{01}$, it becomes \begin{align*} &\lim_{t\to\infty}g_{01}=I, \\&\lim_{t\to\infty}(\hat{\xi}_{1}-\mbox{Ad}_{g_{01}^{-1}}\hat{\xi}_{0})=0,
\end{align*} which is the formation requirement (\ref{eq_for_0i}) with desired configuration $\bar{g}_{01}=I$. That is to say, the follower achieves consensus with the leader.
\end{proof}

Thus, if we are able to design a control law $\hat{u}_{01}$ which can stabilize the relative system $\Sigma_{01}$, then the tracking controller for the follower can be easily obtained by the equation \begin{equation} \hat{u}_{1} = \hat{u}_{01} + \mbox{Ad}_{g_{01}^{-1}}\hat{u}_{0} + [\hat{\xi}_{1},\hat{\xi}_{01}]. \end{equation} Therefore, the tracking problem is converted to designing the stabilization control law $\hat{u}_{01}$ for the relative system $\Sigma_{01}$. Of course, here the control input $\hat{u}_1$ must be designed to be underactuated due to the nonholonomic constraint.

In Lemma \ref{le_sta_nonho}, we have introduced a stabilization control law for the nonholonomic system. If the relative system $\Sigma_{01}$ satisfies the nonholonomic constraint, then $\hat{u}_{01}$ can be directly designed based on Lemma \ref{le_sta_nonho}. Let $\bm{\xi}_{01}=[\omega_{01} \ v_{x01} \ v_{y01}]^{\trans}$ denote the vector form of the relative velocity $\hat{\xi}_{01}$. In the following, we verify whether the side velocity satisfies $v_{y01}=0$. Based on the definition of $\hat{\xi}_{01}$, the term $\mbox{Ad}_{g^{-1}_{01}}\hat{\xi}_0$ should be first calculated, so that we have
\begin{align*} \mbox{Ad}_{g^{-1}_{01}}\hat{\xi}_0 = g_{01}^{-1}\hat{\xi}_{0}g_{01} =  \begin{bmatrix} \hat{\omega}_{0} & R_{01}^{\trans}\left(\hat{\omega}_{0}\bm{r}_{01}+\begin{bmatrix} v_{x0} \\ 0 \end{bmatrix}\right) \\ \bm{0} & 0 \end{bmatrix}, \end{align*} where the vector representing translational motion is \begin{align*} R_{01}^{\trans}\left(\hat{\omega}_{0}\bm{r}_{01}+\begin{bmatrix} v_{x0} \\ 0 \end{bmatrix}\right) &= \begin{bmatrix} \cos\theta_{01} & \sin\theta_{01} \\ -\sin\theta_{01} & \cos\theta_{01} \end{bmatrix} \left(\begin{bmatrix} 0 & -\omega_{0} \\ \omega_{0} & 0 \end{bmatrix}\begin{bmatrix} r_{01x} \\ r_{01y} \end{bmatrix}+\begin{bmatrix} v_{x0} \\ 0 \end{bmatrix}\right) \\ &= \begin{bmatrix} (v_{x0}-\omega_{0}r_{y01})\cos\theta_{01}+\omega_{0}r_{x01}\sin\theta_{01} \\ \omega_{0}r_{x01}\cos\theta_{01}-(v_{x0}-\omega_{0}r_{y01})\sin\theta_{01} \end{bmatrix}. \end{align*} Therefore, the relative velocity $\bm{\xi}_{01}$ is that \begin{align} \label{eq_xi_01} \bm{\xi}_{01} &=\bm{\xi}_1-(\mbox{Ad}_{g_{01}^{-1}}\hat{\xi}_{0})^{\vee} \nonumber \\ &= \begin{bmatrix} \omega_{1}-\omega_{0} \\ v_{x1}-(v_{x0}-\omega_{0}r_{y01})\cos\theta_{01}-\omega_{0}r_{x01}\sin\theta_{01} \\ (v_{x0}-\omega_{0}r_{y01})\sin\theta_{01}-\omega_{0}r_{x01}\cos\theta_{01} \end{bmatrix} \nonumber \\ & = \begin{bmatrix} \omega_{01} & v_{x01} & v_{y01} \end{bmatrix}^{\trans}, \end{align} which indicates the side velocity $$v_{y01}=(v_{x0}-\omega_{0}r_{y01})\sin\theta_{01}-\omega_{0}r_{x01}\cos\theta_{01}\ne0.$$ In other words, the relative system $\Sigma_{01}$ do not satisfy the nonhonomic constraint, with the result that Lemma \ref{le_sta_nonho} cannot be directly used for the stabilization of the relative system $\Sigma_{01}$.

\begin{remark}
As a matter of fact, it is not necessary to make the relative system $\Sigma_{01}$ satisfy the nonholonomic constraint, since the tracking task can be achieved as long as $\Sigma_{01}$ is stabilized to the identity element. However, we still hope that the stabilization of the relative system $\Sigma_{01}$ can be solved by Lemma \ref{le_sta_nonho}, yet which is used for the nonholonomic system particularly. This is because under such an operation the condition of persistency excitation would not be imposed on the leader, which contributes to the tracking control of arbitrary reference trajectories. To be specific, if Lemma \ref{le_sta_nonho} is applied to the stabilization of $\Sigma_{01}$, the tracking task can also be achieved in the case where the leader's velocity is equivalent to 0. The reason is that when the velocity of the leader becomes 0, the tracking problem would naturally degenerate to the stabilization, while the achievement of which is guaranteed by Lemma \ref{le_sta_nonho} completely. Thus, under such a situation, for any type of the reference trajectory, the follower is able to track all of them.
\end{remark}

Based on above analysis, in the following we need to make the relative system $\Sigma_{01}$ nonholonomic constrained so as to use the stabilization control in Lemma \ref{le_sta_nonho}.

According to the relative velocity $\bm{\xi}_{01}$ in (\ref{eq_xi_01}), let the side velocity $v_{y01}=0$, we obtain that \begin{equation} \label{eq_vy01_0} (v_{x0}-\omega_{0}r_{y01})\sin\theta_{01}-\omega_{0}r_{x01}\cos\theta_{01}=0. \end{equation} The configuration that satisfies (\ref{eq_vy01_0}) is named the adjoint orbit \cite{LiuND}. Now, define the following auxiliary attitude angle \begin{equation} \label{eq_til_theta} \tilde{\theta}_{01}=\arctan\frac{\omega_{0}r_{x01}}{v_{x0}-\omega_{0}r_{y01}}, \end{equation} which is called the adjoint attitude angle \cite{TayefiIJC}. Define the following auxiliary configuration, or adjoint configuration \begin{equation} \label{eq_til_g1} \tilde{g}_{1}=\begin{bmatrix} R(\tilde{\theta}_{1}) & \bm{p}_{1} \\ \bm{0} & 1 \end{bmatrix}, \end{equation} where $\tilde{\theta}_{1}=\theta_{0}+\tilde{\theta}_{01}$. From the definition in (\ref{eq_til_g1}), it can be observed that the auxiliary configuration $\tilde{g}_1$ has same position as follower $g_{1}$, while its orientation is decided by the leader's attitude $\theta_{0}$ and the adjoint attitude $\tilde{\theta}_{01}$.

In the following, we shall define two new relative subsystems, and establish the relationship between the relative system $\Sigma_{01}$ and two relative subsystems.

Define the relative configuration \begin{equation} \label{eq_til_g01} \tilde{g}_{01}=g_{0}^{-1}\tilde{g}_{1}, \end{equation} which represents the relative configuration of the auxiliary system $\tilde{g}_{1}$ with respect to the leader $g_{0}$. Let $\tilde{\bm{\xi}}_{1}=[ \tilde{\omega}_{1} \ \ \tilde{v}_{x1} \ \ \tilde{v}_{y1} ]^{\trans}$ denote the velocity of the auxiliary system, and $\tilde{\bm{u}}_{1}$ denotes the control input. With the similar derivation of the dynamics of $g_{01}$ in (\ref{eq_sys01}), we obtain the dynamics of $\tilde{g}_{01}$ as \begin{equation} \label{eq_til_sys01} \tilde{\Sigma}_{01}: \left\{ \begin{aligned} & \dot{\tilde{g}}_{01}=\tilde{g}_{01}\hat{\tilde{\xi}}_{01} \\ & \dot{\hat{\tilde{\xi}}}_{01}=\hat{\tilde{u}}_{01} \end{aligned}, \right. \end{equation}
where the relative velocity is \begin{equation} \label{eq_til_xi01} \hat{\tilde{\xi}}_{01} = \hat{\tilde{\xi}}_{1} - \mbox{Ad}_{\tilde{g}_{01}^{-1}} \hat{\xi}_{0}, \end{equation} and the relative control input is \begin{equation} \label{eq_til_u01} \hat{\tilde{u}}_{01} = \hat{\tilde{u}}_{1} - \mbox{Ad}_{\tilde{g}_{01}^{-1}}\hat{u}_{0} - [\hat{\tilde{\xi}}_{1},\hat{\tilde{\xi}}_{01}]. \end{equation}

\begin{lemma}
\label{le_til_g01_non}
If the side velocity of the auxiliary system is equivalent to zero, i.e. $\tilde{v}_{y1}=0$, then the relative system $\tilde{\Sigma}_{01}$ is nonholonomic constrained.
\end{lemma}

\begin{proof}
According to the definition in (\ref{eq_til_g01}), $\tilde{g}_{01}$ can be calculated as \begin{equation*} \tilde{g}_{01} = \begin{bmatrix} R_{0}^{\trans}R(\tilde{\theta}_{1}) & R_{0}^{\trans}(\bm{p}_{1}-\bm{p}_{0}) \\ \bm{0} & 1 \end{bmatrix} \triangleq \begin{bmatrix} R(\tilde{\theta}_{01}) & \tilde{\bm{r}}_{01} \\ \bm{0} & 1 \end{bmatrix}, \end{equation*} where the rotation matrix is \begin{equation*} R(\tilde{\theta}_{01})=\begin{bmatrix} \cos\tilde{\theta}_{01} & -\sin\tilde{\theta}_{01} \\ \sin\tilde{\theta}_{01} & \cos\tilde{\theta}_{01} \end{bmatrix}, \end{equation*} and the position vector is $\tilde{\bm{r}}_{01}=[\tilde{r}_{x01}\ \ \tilde{r}_{y01}]^{\trans}$. Compared with the definition of $g_{01}$ in (\ref{eq_g01}), it can be obtained that $\bm{r}_{01}=\tilde{\bm{r}}_{01}$, that is, \begin{equation} \label{eq_r_e_tr} [r_{x01}\ \ r_{y01}]^{\trans}=[\tilde{r}_{x01}\ \ \tilde{r}_{y01}]^{\trans}. \end{equation} Based on (\ref{eq_til_xi01}), the vector form of the relative velocity $\tilde{\bm{\xi}}_{01}$ is \begin{align} \tilde{\bm{\xi}}_{01} &= \begin{bmatrix} \tilde{\omega}_{1}-\omega_{0} \\ \tilde{v}_{x1}-(v_{x0}-\omega_{0}\tilde{r}_{y01})\cos\tilde{\theta}_{01}-\omega_{0}\tilde{r}_{x01}\sin\tilde{\theta}_{01} \\ \tilde{v}_{y1}+(v_{x0}-\omega_{0}\tilde{r}_{y01})\sin\tilde{\theta}_{01}-\omega_{0}\tilde{r}_{x01}\cos\tilde{\theta}_{01} \end{bmatrix} \nonumber \\ &\triangleq [\tilde{\omega}_{01}\ \ \tilde{v}_{x01}\ \ \tilde{v}_{y01}]^{\trans} \end{align} Substitute (\ref{eq_til_theta}) and (\ref{eq_r_e_tr}) into $\tilde{v}_{y01}$, and employ the condition $\tilde{v}_{y1}=0$, we have \begin{equation*} \tilde{v}_{y01}=0, \end{equation*} which indicates that the relative velocity $\tilde{\bm{\xi}}_{01}$ satisfies the nonholonomic constraint.
\end{proof}

Having provided the definition of $\tilde{g}_{01}$, we define another relative configuration \begin{equation} \label{eq_defi_ge} g_{e}=\tilde{g}_{1}^{-1}g_{1}, \end{equation} which represents the relative configuration of the follower $g_{1}$ with respect to the auxiliary system $\tilde{g}_{1}$. The dynamics of $g_{e}$ can be derived as \begin{equation} \label{eq_sys_e} \Sigma_{e}:\left\{ \begin{aligned} & \dot{g}_{e}={g}_{e}\hat{\xi}_{e}, \\ & \dot{\hat{\xi}}_{e}=\hat{u}_{e}, \end{aligned}, \right. \end{equation} where the relative velocity is \begin{equation} \label{eq_xie} \hat{\xi}_{e}=\hat{\xi}_{1}-\mbox{Ad}_{g_{e}^{-1}}\hat{\tilde{\xi}}_{1}, \end{equation} and the relative control input is \begin{equation} \label{eq_ue} \hat{u}_{e}=\hat{u}_{1}-\mbox{Ad}_{g_{e}^{-1}}\hat{\tilde{u}}_{1}-[\hat{\xi}_{1},\hat{\xi}_{e}]. \end{equation}

Up to now, we have defined two relative subsystems $\tilde{\Sigma}_{01}$ and $\Sigma_{e}$, the relationship between which and the relative system $\Sigma_{01}$ is revealed in the following Lemma.

\begin{lemma}[Decomposition Lemma]
\label{le_two_sys}
The relative system $\Sigma_{01}$ is stabilized to the identity element, if two relative subsystems $\tilde{\Sigma}_{01}$ and $\Sigma_{e}$ are both stabilized to the identity element.
\end{lemma}

\begin{proof}
Based on the definition of $g_{01}$, there holds \begin{equation} \label{eq_tilg_ge} g_{01} = g_{0}^{-1}g_{1} = g_{0}^{-1}\tilde{g}_{1}\tilde{g}_{1}^{-1}g_{1}=\tilde{g}_{01}g_{e}. \end{equation} From (\ref{eq_til_xi01}), we obtain $\hat{\tilde{\xi}}_{1}$ as \begin{equation} \label{eq_til_xi_1} \hat{\tilde{\xi}}_{1} = \hat{\tilde{\xi}}_{01} + \mbox{Ad}_{\tilde{g}_{01}^{-1}} \hat{\xi}_{0}. \end{equation} Left multiple $\mbox{Ad}_{g_{e}^{-1}}$ to both sides of (\ref{eq_til_xi_1}) and substitute the result into (\ref{eq_xie}), there holds \begin{equation} \label{eq_le_bet_m} \hat{\xi}_{e} = \hat{\xi}_{1}-\mbox{Ad}_{g_{e}^{-1}}\hat{\tilde{\xi}}_{01} - \mbox{Ad}_{g_{e}^{-1}}\mbox{Ad}_{\tilde{g}_{01}^{-1}}\hat{\xi}_{0}. \end{equation} With the following property \begin{equation*}\mbox{Ad}_{g_{e}^{-1}}\mbox{Ad}_{\tilde{g}_{01}^{-1}}\hat{\xi}_{0} = \mbox{Ad}_{(\tilde{g}_{01}g_{e})^{-1}}\hat{\xi}_{0} = \mbox{Ad}_{g_{01}^{-1}}\hat{\xi}_{0},\end{equation*} equation (\ref{eq_le_bet_m}) can be rewritten as \begin{equation*} \hat{\xi}_{e}+\mbox{Ad}_{g_{e}^{-1}}\hat{\tilde{\xi}}_{01}=\hat{\xi}_{1}-\mbox{Ad}_{g_{01}^{-1}}\hat{\xi}_{0}, \end{equation*} that is \begin{equation} \label{eq_tilxi_xie} \hat{\xi}_{01}=\hat{\xi}_{e}+\mbox{Ad}_{g_{e}^{-1}}\hat{\tilde{\xi}}_{01}, \end{equation} where (\ref{eq_xi_01_0}) is employed. When the relative subsystems $\tilde{\Sigma}_{01}$ and $\Sigma_{e}$ are both stabilized to the identity element, there holds \begin{equation*} \begin{aligned} &\lim_{t\to\infty}\tilde{g}_{01}=I \\ &\lim_{t\to\infty}\hat{\tilde{\xi}}_{01}=0 \end{aligned} \quad \mbox{and} \quad \begin{aligned} &\lim_{t\to\infty}g_{e}=I \\ &\lim_{t\to\infty}\hat{\xi}_{e}=0. \end{aligned}\end{equation*} Thus, according to (\ref{eq_tilg_ge}) and (\ref{eq_tilxi_xie}), the limits of $g_{01}$ and $\hat{\xi}_{01}$ can be computed as \begin{align*} &\lim_{t\to\infty}g_{01}=\lim_{t\to\infty}\tilde{g}_{01}\lim_{t\to\infty}g_{e}=I, \\ &\lim_{t\to\infty}\hat{\xi}_{01}=\lim_{t\to\infty}\hat{\xi}_{e}+\mbox{Ad}_{g_{e}^{-1}}(\lim_{t\to\infty}\hat{\tilde{\xi}}_{01})=0, \end{align*} which indicates that the relative system $\Sigma_{01}$ is stabilized to the identity element.
\end{proof}

Therefore, the tracking problem is converted to the stabilization for two relative subsystems $\tilde{\Sigma}_{01}$ and $\Sigma_{e}$. Since the task is to design the controller for the follower, it is necessary to derive the expression of the control input $\hat{u}_{1}$. From (\ref{eq_til_u01}), the control input for the auxiliary system is obtained as \begin{equation} \hat{\tilde{u}}_{1} = \hat{\tilde{u}}_{01} + \mbox{Ad}_{\tilde{g}_{01}^{-1}}\hat{u}_{0} + [\hat{\tilde{\xi}}_{1},\hat{\tilde{\xi}}_{01}]. \end{equation} Substitute it into (\ref{eq_ue}), it can be followed that the control input of the follower has the following form \begin{equation} \label{eq_u1_v1} \hat{u}_{1}=\hat{u}_{e} + \mbox{Ad}_{g_{e}^{-1}}(\hat{\tilde{u}}_{01} + \mbox{Ad}_{\tilde{g}_{01}^{-1}}\hat{u}_{0} + [\hat{\tilde{\xi}}_{1},\hat{\tilde{\xi}}_{01}]) +[\hat{\xi}_{1},\hat{\xi}_{e}]. \end{equation} Thus, as long as we design the stabilization control laws $\hat{\tilde{u}}_{01}$ and $\hat{u}_{e}$ and substitute them into (\ref{eq_u1_v1}), then the follower's tracking control law $\hat{u}_{1}$ can be obtained naturally. In the following, we shall design the control inputs $\hat{\tilde{u}}_{01}$ and $\hat{u}_{e}$ for stabilization.

\vspace{1ex}
\emph{1) Design for $\hat{u}_{e}$}
\vspace{1ex}

Based on the definition in (\ref{eq_defi_ge}), the matrix expression of $g_{e}$ is \begin{equation*} g_{e}=\begin{bmatrix} \tilde{R}_{1}^{\trans} & -\tilde{R}_{1}^{\trans}\bm{p}_{1} \\ \bm{0} & 1 \end{bmatrix} \begin{bmatrix} R_{1} & \bm{p}_{1} \\ \bm{0} & 1 \end{bmatrix} = \begin{bmatrix} R(\theta_{1}-\tilde{\theta}_{1}) & 0 \\ \bm{0} & 1 \end{bmatrix}. \end{equation*} Design the velocity of the auxiliary system as \begin{equation*} \hat{\tilde{\xi}}_{1}=\mbox{Ad}_{g_{e}}\hat{\xi}_{1}, \end{equation*} and substitute it into (\ref{eq_xie}), so that the relative velocity becomes \begin{equation} \label{eq_xie_e0} \hat{\xi}_{e}=\hat{\xi}_{1}-\mbox{Ad}_{g_{e}^{-1}}\mbox{Ad}_{g_{e}}\hat{\xi}_{1}=0. \end{equation} Now, we introduce a Lemma to design the stabilization control input $\hat{u}_{e}$.

\begin{lemma}[\cite{Bullo1995}, Theorem 6]
\label{le_sta_ful}
For the fully-actuated system evolving in SE(2) with dynamics (\ref{sys}), the system state $g$ can be stabilized to the identity element from any initial condition with $\mbox{tr}(g)\ne-1$ under the following control law \begin{equation*} \hat{u} = -k_p\log_{\begin{scriptsize}\mbox{SE(2)}\end{scriptsize}}(g) - k_d\hat{\xi}, \end{equation*} where $k_p$ and $k_d$ are positive control gains.
\end{lemma}

According to Lemma \ref{le_sta_ful}, and employing the property in (\ref{eq_xie_e0}), the stabilization control input $\hat{u}_{e}$ can be designed as \begin{equation*} \hat{u}_{e}=-k_e\log_{\begin{scriptsize}\mbox{SE(2)}\end{scriptsize}}(g_{e})-\hat{\xi}_{e}= -k_e\log_{\begin{scriptsize}\mbox{SE(2)}\end{scriptsize}}(g_{e}), \end{equation*} where $k_e$ is a positive control gain. By the definition of the logarithm map $ \log_{\begin{scriptsize}\mbox{SE(2)}\end{scriptsize}}$, the vector form of $\hat{u}_{e}$ is that \begin{equation} \label{eq_ue_sta} \bm{u}_{e}=-k_e\begin{bmatrix} \theta_{1}-\tilde{\theta}_{1} \\ 0 \\ 0 \end{bmatrix}. \end{equation}

\emph{2) Design for $\hat{\tilde{u}}_{01}$}
\vspace{1ex}

Since control input $u_{e}$ in (\ref{eq_ue_sta}) can stabilize the relative system $\Sigma_{e}$, it is able to obtain that $\tilde{g}_{1}\to g_{1}$ and $\tilde{\xi}_{1}\to\xi_{1}$, resulting in that the side velocity of the auxiliary system satisfying $\tilde{v}_{y1}\to v_{y1}=0$. Then, according to Lemma \ref{le_til_g01_non}, the relative system $\tilde{\Sigma}_{01}$ becomes nonholonomic constrained. By dint of Lemma \ref{le_sta_nonho}, the stabilization control input $\hat{\tilde{u}}_{01}$ for the nonholonomic system $\tilde{\Sigma}_{01}$ can be designed in the vector form as \begin{equation*} \tilde{u}_{01}=-k_p\begin{bmatrix} \tilde{\theta}_{01}+k\tilde{\beta}_{01} \\ \tilde{q}_{x01} \\ 0 \end{bmatrix} -k_d\begin{bmatrix} \tilde{\omega}_{01} \\ \tilde{v}_{x01} \\ 0 \end{bmatrix}, \end{equation*} where $\tilde{\beta}_{01}=-\arctan(\tilde{q}_{y01}/ \tilde{q}_{x01})$, $\tilde{\bm{q}}=[\tilde{q}_{x01} \ \ \tilde{q}_{y01} ]^{\trans} =A^{-1}(\tilde{\theta}_{01})\tilde{\bm{r}}_{01}$, and $k_p$, $k_d$, $k$ are all scalar control gains.

Up to now, we have designed the stabilization control law $\hat{u}_{e}$ and $\hat{\tilde{u}}_{01}$. With the help of $u_{e}$, there holds that $g_{e}\to I$ and $\tilde{g}_{01}\to g_{01}$. Thus, $\tilde{u}_{01}$ can be rewritten as \begin{equation} \label{eq_til_u01_sta} \tilde{u}_{01}=-k_p\begin{bmatrix} \theta_{01}+k\beta_{01} \\ q_{x01} \\ 0 \end{bmatrix} -k_d\begin{bmatrix} \omega_{01} \\ v_{x01} \\ 0 \end{bmatrix}, \end{equation} where $\beta_{01}=-\arctan(q_{y01}/ q_{x01})$, $\bm{q}=[q_{x01} \ \ q_{y01} ]^{\trans} =A^{-1}(\theta_{01})\bm{r}_{01}$, and the follower's control input $\hat{u}_{1}$ in (\ref{eq_u1_v1}) can be rewritten as \begin{equation} \label{eq_u1_v2} \hat{u}_{1} = \hat{u}_{e} + \hat{\tilde{u}}_{01} + \mbox{Ad}_{g_{01}^{-1}}\hat{u}_{0} + [\hat{\xi}_{1},\hat{\xi}_{01}]. \end{equation} Therefore, the tracking control law can be obtained by substituting (\ref{eq_ue_sta}) and (\ref{eq_til_u01_sta}) into (\ref{eq_u1_v2}), and after computation we acquire the vector form of the tracking control law $\bm{u}_{1}$, whose components in rotation and translation are \begin{equation} \label{eq_u1_v3} \begin{aligned}u_{\theta1}=&-k_{e}(\theta_{1}-\tilde{\theta}_{1})-k_p(\theta_{01}+k\beta_{01})-k_d\omega_{01}+u_{\theta 0}, \\ u_{x1}=&-k_pq_{x01}-k_dv_{x01}+(u_{x0}-u_{\theta 0 }r_{y01})\cos\theta_{01}+u_{\theta 0}r_{x01}\sin\theta_{01}. \end{aligned} \end{equation}
We summarize the above results in the following Theorem.

\begin{theorem}[Single Follower Tracking]
\label{the_sing_tra}
Consider two nonholonomic mobile robots described by dynamics (\ref{sys}), which are connected by a directed edge. For arbitrary reference trajectory of the leader, if the control strategy is designed as (\ref{eq_u1_v3}), then the follower is able to track the trajectory of the leader globally and asymptotically.
\end{theorem}

\begin{proof}
The tracking control law (\ref{eq_u1_v3}) is derived from the control laws (\ref{eq_ue_sta}) and (\ref{eq_til_u01_sta}), which can stabilize two relative subsystems $\Sigma_{e}$ and $\tilde{\Sigma}_{01}$ to the identity element. According to Lemma \ref{le_two_sys}, this is equivalent to the stabilization of the relative system $\Sigma_{01}$, which guarantees the tracking of the follower with respect to the leader by dint of Lemma \ref{le_sys_01}. Furthermore, the global and asymptotical convergence of the tracking controller can be ensured by the stabilization control laws (\ref{eq_ue_sta}) and (\ref{eq_til_u01_sta}) both globally and asymptotically stable.
\end{proof}

\subsection{Consensus Tracking}

In this subsection, we investigate the problem of consensus tracking for multiple nonholonomic mobile robots, under the assumption that the communication topology among them is given by a directed acyclic graph with one root node. Compared with the directed spanning tree (DST), the directed acyclic graph (DAG) are more general and reliable, since each child node can acquire information from more than one parent nodes. However, to design the consensus tracking control law under DAG is more difficult and involved than that under DST. In the communication topology of DST, each child node has only one parent node, in other words, each follower only has one leader of its own, so that the consensus tracking problem can be converted into one leader and one follower tracking problem. In contrast, under DAG communication topology, each child node may obtain information from more than parent nodes, that is to say, the follower has multiple leaders, which brings more difficulties for the tracking control law design.

Thus, the essential problem of consensus tracking under DAG is that one follower tracks multiple leaders of its own. Herein, the ``leader" does not mean the root node, but the parent node of some robot. In the following, for clear illustration, we employ ``global leader" to describe the root vertex of the whole network and ``local leader" to represent the parent vertex of a certain robot. Therefore, the problem is to design a controller for the follower to track multiple local leaders. In order to employ the result of single follower tracking, we introduce the convex combination in nonlinear manifolds \cite{Peng} to construct a virtual local leader, so that one follower tracking multiple local leaders can be converted to tracking only one virtual local leader.

For any nonholonomic mobile robot $i$, let $g_{i}^{1},g_{i}^{2},\cdots,g_{i}^{M_{i}}$ denote the configuration of its local leaders, where $M_{i}$ is the number of the local leaders. The convex combination of $g_{i}^{1},g_{i}^{2},\cdots,g_{i}^{M_{i}}$ is denoted by $g_{c_{i}}$, which is iteratively defined as \begin{equation} \label{eq_g_c} \begin{aligned} & g_{i}^{1,2}=g_{i}^{1}\exp(\lambda_{i}^{1}(\log((g_{i}^{1})^{-1}g_{i}^{2}))), \\ &
g_{c_{i}}=g_{i}^{1,\cdots,M_{i}-1}\exp(\lambda_{i}^{M_{i}-1}(\log((g_{i}^{1,\cdots,M_{i}-1})^{-1}g_{i}^{M_{i}}))),
\end{aligned}\end{equation} where $\lambda_{i}$ is the convex combination coefficient and satisfies $0\leq\lambda_{i}\leq1$, $i=1,2,\cdots,M_{i}-1$. Let $\xi_{i}^{1},\xi_{i}^{2},\cdots,\xi_{i}^{M_{i}}$ denote the velocity of the local leaders of the robot $i$. Then, the convex combination of $\xi_{i}^{1},\xi_{i}^{2},\cdots,\xi_{i}^{M_{i}}$, represented by $\xi_{c_{i}}$, is defined as \begin{equation} \label{eq_xi_c} \begin{aligned} & \xi_{i}^{1,2}=(1-\lambda_{i}^{1})\xi_{i}^{1}+\lambda_{i}^{1}\xi_{i}^{2}, \\ & \xi_{c_{i}}=(1-\lambda_{i}^{M_{i}-1})\xi_{i}^{1,\cdots,M_{i}-1}+\lambda_{i}^{M_{i}-1}\xi_{i}^{M_{i}}. \end{aligned} \end{equation} Similarly, let $u_{i}^{1},u_{i}^{2},\cdots,u_{i}^{M_{i}}$ denote the control input the local leaders of robot $i$, and the definition of their convex combination $u_{c_{i}}$ is \begin{equation} \label{eq_u_c} \begin{aligned} & u_{i}^{1,2}=(1-\lambda_{i}^{1})u_{i}^{1}+\lambda_{i}^{1}u_{i}^{2}, \\ & u_{c_{i}}=(1-\lambda_{i}^{M_{i}-1})u_{i}^{1,\cdots,M_{i}-1}+\lambda_{i}^{M_{i}-1}u_{i}^{M_{i}}. \end{aligned} \end{equation} More information about the convex combination in nonlinear manifolds can be founded in \cite{Peng}.

\begin{lemma}[\cite{Peng}, Corollary 3.2]
The dynamics of the convex combination configuration $g_{c_{i}}$, defined in (\ref{eq_g_c}), is \begin{equation} \label{eq_sys_c} \Sigma_{c_{i}}: \left\{ \begin{aligned} & \dot{g}_{c_{i}}=g_{c_{i}}\hat{\xi}_{c_{i}} \\ & \hat{\xi}_{c_{i}}=\hat{u}_{c_{i}} \end{aligned}, \right. \end{equation} where the definitions of velocity $\xi_{c_{i}}$ and control input $u_{c_{i}}$ are defined in (\ref{eq_xi_c}) and (\ref{eq_u_c}) respectively.
\end{lemma}

Note that it has been proved in \cite{TayefiND} that the convex combination of nonholonomic mobile robots still satisfies nonholonomic constraint. That is to say, the system $\Sigma_{c_{i}}$ is a nonholonomic system. Thus, for any robot $i$ ($i=1,2,\cdots,N$), we can construct a virtual local leader with the dynamics (\ref{eq_sys_c}). With the aid of such a virtual system, the consensus tracking of the whole vehicles with the global leader can be converted to the tracking of each vehicle with its virtual local leader. The following Lemma explains this relationship in detail.

\begin{lemma}[\cite{TayefiND}, Corollary 1]
\label{le_convex_consen_tra}
Consider a networked system of $N+1$ nonholonomic mobile robots in SE(2), which are connected by a directed acyclic graph with one root node. All the robots achieve consensus with the root node (or global leader), i.e. \begin{equation*}\lim_{t\to\infty}g_{i}=g_{0},\ i=1,2,\cdots,N, \end{equation*} if each robot $i$ track the convex combination of its parent nodes (or local leaders), i.e. \begin{equation*}\lim_{t\to\infty}g_{i}=g_{c_{i}},\ i=1,2,\cdots,N. \end{equation*}
\end{lemma}

In the following, based on Lemma \ref{le_convex_consen_tra}, we derive the consensus tracking control law for the robot $i$. Design the relative configuration of $g_{i}$ with respect to $g_{c_{i}}$, denoted by $g_{c_{i}i}$, as \begin{equation} g_{c_{i}i}=g_{c_{i}}^{-1}g_{i} \triangleq \begin{bmatrix} R(\theta_{c_{i}i}) & \bm{r}_{c_{i}i} \\ 0 & 1 \end{bmatrix}, \end{equation} and the exponential coordinates of $g_{c_{i}i}$ is \begin{equation} \bm{X}_{c_{i}i} = \log_{\begin{scriptsize}\mbox{SE(2)}\end{scriptsize}}(g_{c_{i}i})\triangleq\begin{bmatrix} \theta_{c_{i}i} & q_{xc_{i}i} & q_{yc_{i}i}\end{bmatrix}^{\trans}.\end{equation} Define the relative velocity corresponding to $g_{c_{i}i}$ \begin{equation} \bm{\xi}_{c_{i}i}=\bm{\xi}_{i}-(\mbox{Ad}_{g_{c_{i}i}^{-1}}\hat{\xi}_{c_{i}i})^{\vee} \triangleq \begin{bmatrix} \omega_{c_{i}i} & v_{xc_{i}i} & v_{yc_{i}i} \end{bmatrix}^{\trans}. \end{equation} Similar to (\ref{eq_til_theta}), we design the following adjoint attitude \begin{equation} \label{eq_til_theta_i} \tilde{\theta}_{c_{i}i}=\arctan\frac{\omega_{c_{i}}r_{xc_{i}i}}{v_{xc_{i}}-\omega_{c_{i}}r_{yc_{i}i}}, \end{equation} where $\omega_{c_{i}}$ and $v_{xc_{i}}$ are the components of the convex combination velocity vector $\bm{\xi}_{c_{i}}$, and $r_{xc_{i}i}$ and $r_{yc_{i}i}$ are the components of the relative position vector $\bm{r}_{c_{i}i}$. Analogous to the single follower case, we can design the following control input for robot $i$, that is \begin{equation} \label{eq_ui_con_tr_0} \begin{aligned}u_{\theta i}=&-k_{e}(\theta_{i}-\tilde{\theta}_{i})-k_p(\theta_{c_{i}i}+k\beta_{c_{i}i})-k_d\omega_{c_{i}i}+u_{\theta c_{i}},\\ u_{x i}=&-k_pq_{xc_{i}i}-k_dv_{xc_{i}i}+(u_{xc_{i}}-u_{\theta c_{i} }r_{yc_{i}i})\cos\theta_{c_{i}i} +u_{\theta c_{i}}r_{xc_{i}i}\sin\theta_{c_{i}i}, \end{aligned} \end{equation} in which $\tilde{\theta}_{i}=\theta_{c_{i}}+\tilde{\theta}_{c_{i}i}$, $\beta_{c_{i}i}=-\arctan(q_{yc_{i}i}/q_{xc_{i}i})$, $u_{\theta c_{i}}$ and $u_{x c_{i}}$ are the components of the convex combination control input $\bm{u}_{c_{i}}$, $i=1,2,\cdots,N$. The Theorem of consensus tracking control can be presented as follows.

\begin{theorem}[Consensus Tracking]
\label{the_multi_tra}
Consider $N+1$ nonholonomic mobile robots with dynamics (\ref{sys}), which are connected by a directed acyclic graph with one root node. For arbitrary reference trajectory of the leader, if the control strategy is designed as (\ref{eq_ui_con_tr_0}), then all the followers are able to achieve consensus tracking with the leader globally and asymptotically.
\end{theorem}

\begin{proof}
According to Theorem \ref{the_sing_tra}, the control law (\ref{eq_ui_con_tr_0}) can drive the robot $i$ asymptotically track the convex combination of its local leaders (or parent nodes), that is \begin{equation*} \lim_{t\to\infty}g_{i}=g_{c_{i}},\ i=1,2,\cdots,N. \end{equation*} Then, with the aid of Lemma \ref{le_convex_consen_tra}, the consensus tracking of all the followers with the leader can be achieved under such a controller.
\end{proof}

\subsection{Formation Tracking}

We have obtained the consensus tracking control law in the previous subsection. If the control problem of formation tracking can be converted to that of consensus tracking, then the above presented result can be applied to the formation tracking control, which greatly simplifies the control design process. Hence, in the following, we introduce a new transformed system and reveal the relationship between consensus control and formation control.

We firstly consider the formation tracking control of one leader and one follower. The desired relative configuration of the follower with respect to the leader is denoted by $\bar{g}_{01}$, that is \begin{equation} \bar{g}_{01}=\begin{bmatrix} R(\bar{\theta}_{01}) & \bar{\bm{p}}_{01} \\ \bm{0} & 1\end{bmatrix}, \end{equation} where $\bar{\theta}_{01}$ is the desired relative attitude angle, and $\bar{\bm{p}}_{01}=[\bar{x}_{01}\ \ \bar{y}_{01}]^{\trans}$ is the desired relative position vector. Due to the nonholonomic constraint, $\bar{\theta}_{01}$ and $\bar{p}_{01}$ are not independent each other. It has been proved in \cite{TayefiND} that for arbitrary relative position vector $\bar{\bm{p}}_{01}$, the relative attitude angle $\bar{\theta}_{01}$ is decided by the equality \begin{equation} \label{eq_bar_the} \bar{\theta}_{01} = \arctan\frac{\omega_{0}\bar{x}_{01}}{v_{x0}-\omega_{0}\bar{y}_{01}}. \end{equation}

Now, define a new transformed configuration \begin{equation} g_{a_{1}}=g_{0}\bar{g}_{01}. \end{equation} Take the derivative of $g_{a_{1}}$ with respect to time, and we have \begin{equation*} \dot{g}_{a_{1}} = \dot{g}_{0}\bar{g}_{01} = g_{0}\hat{\xi}_{0}\bar{g}_{01} = g_{a_{1}}\mbox{Ad}_{\bar{g}_{01}^{-1}}\hat{\xi}_{0}. \end{equation*} Define the transformed velocity $\hat{\xi}_{a_{1}} = \mbox{Ad}_{\bar{g}_{01}^{-1}}\hat{\xi}_{0}$, and its time derivative is $\dot{\hat{\xi}}_{a_{1}}=\mbox{Ad}_{\bar{g}_{01}^{-1}}\dot{\hat{\xi}}_{0} = \mbox{Ad}_{\bar{g}_{01}^{-1}}\hat{u}_{0}$. Thus, the dynamics of the transformed system is \begin{equation} \Sigma_{a_{1}}: \left\{ \begin{aligned} & \dot{g}_{a_{1}}=g_{a_{1}}\hat{\xi}_{a_{1}} \\ & \dot{\hat{\xi}}_{a_{1}}=\hat{u}_{a_{1}} \end{aligned}, \right.  \end{equation} where the transformed control input is $\hat{u}_{a_{1}}=\mbox{Ad}_{\bar{g}_{01}^{-1}}\hat{u}_{0}$.

\begin{lemma}
\label{le_for_con_rela}
The follower $\Sigma_{1}$ achieves the formation $\bar{g}_{01}$ with respect to the leader $\Sigma_{0}$, if $\Sigma_{1}$ tracks the transformed system $\Sigma_{a_{1}}$.
\end{lemma}

\begin{proof}
The tracking of $\Sigma_{1}$ with $\Sigma_{a_{1}}$ implies \begin{subequations} \begin{align} \label{eq_ga1_for_1}& \lim_{t\to\infty}g_{a_{1}}^{-1}g_{1}=I \\ \label{eq_xia1_for_1} & \lim_{t\to\infty}(\hat{\xi}_{1}^{s}-\hat{\xi}_{a_{1}}^{s})=0 \end{align} \end{subequations} where $g_{a_{1}}^{-1}g_{1}$ is the relative configuration of $\Sigma_{1}$ with respect to $\Sigma_{a_{1}}$, $\hat{\xi}_{1}^{s}$ and $\hat{\xi}_{a_{1}}^{s}$ are velocities of two systems in spacial frame. According to the definition of $g_{a_{1}}$, there holds \begin{equation*} g_{a_{1}}^{-1}g_{1}=\bar{g}_{01}^{-1}g_{0}^{-1}g_{1}=\bar{g}_{01}^{-1}g_{01}. \end{equation*} Thus, with further computation, (\ref{eq_ga1_for_1}) becomes \begin{equation} \label{eq_ga1_for_2} \lim_{t\to\infty}g_{01}=\bar{g}_{01}. \end{equation} From the definition of spacial velocity $\hat{\xi}_{1}^{s}$ and $\hat{\xi}_{a_{1}}^{s}$, it follows \begin{align*} \hat{\xi}_{1}^{s}-\hat{\xi}_{a_{1}}^{s} &= \mbox{Ad}_{g_{1}}\hat{\xi}_{1}-\mbox{Ad}_{g_{a_{1}}}\hat{\xi}_{a_{1}} \\ &=\mbox{Ad}_{g_{1}}(\hat{\xi}_{1}-\mbox{Ad}_{g_{1}^{-1}g_{0}\bar{g}_{01}}\hat{\xi}_{a_{1}}). \end{align*} Substitute $\hat{\xi}_{a_{1}} = \mbox{Ad}_{\bar{g}_{01}^{-1}}\hat{\xi}_{0}$ into above equation, we have \begin{align*} \hat{\xi}_{1}^{s}-\hat{\xi}_{a_{1}}^{s} &= \mbox{Ad}_{g_{1}}(\hat{\xi}_{1}-\mbox{Ad}_{g_{01}^{-1}\bar{g}_{01}}\mbox{Ad}_{\bar{g}_{01}^{-1}}\hat{\xi}_{0}) \\ & = \mbox{Ad}_{g_{1}}(\hat{\xi}_{1}-\mbox{Ad}_{g_{01}^{-1}}\hat{\xi}_{0}). \end{align*} Based on (\ref{eq_xia1_for_1}), it is obtained \begin{equation} \label{eq_xia1_for_2} \lim_{t\to\infty}(\hat{\xi}_{1}-\mbox{Ad}_{g_{01}^{-1}}\hat{\xi}_{0})=0. \end{equation} Therefore, (\ref{eq_ga1_for_2}) and (\ref{eq_xia1_for_2}) suggest that the leader $\Sigma_{0}$ and the follower $\Sigma_{1}$ achieve the formation specified by $\bar{g}_{01}$.
\end{proof}

The formation problem is converted to the tracking problem employing Lemma \ref{le_for_con_rela}. The follower $\Sigma_{1}$ satisfies the nonholonomic constraint naturally. If the transformed system $\Sigma_{a_{1}}$ is also nonholonomic constrained, the controller (\ref{eq_u1_v3}) can be directly applied to the tracking problem of $\Sigma_{1}$ with $\Sigma_{a_{1}}$. Based on the definition of $\hat{\xi}_{a_{1}}=\mbox{Ad}_{\bar{g}_{01}^{-1}}\hat{\xi}_{0}$, the vector form of the transformed velocity is \begin{equation*} \bm{\xi}_{a_{1}}=\begin{bmatrix} \omega_{0} \\ (v_{x0}-\omega_{0}\bar{y}_{01})\cos\bar{\theta}_{01}+\omega_{0}\bar{x}_{01}\sin\bar{\theta}_{01} \\ -(v_{x0}-\omega_{0}\bar{y}_{01})\sin\bar{\theta}_{01}+\omega_{0}\bar{x}_{01}\cos\bar{\theta}_{01} \end{bmatrix} \triangleq \begin{bmatrix} \omega_{a_{1}} \\ v_{xa_{1}} \\ v_{ya_{1}} \end{bmatrix}. \end{equation*} Because $\bar{\theta}_{01}$ is determined by $\bar{x}_{01}$ and $\bar{y}_{01}$, by substituting (\ref{eq_bar_the}) into $v_{ya_{1}}$ it follows that $v_{ya_{1}}=0$, which means that the transformed system $\Sigma_{a_{1}}$ satisfies the nonholonomic constraint and that Theorem \ref{the_sing_tra} can be used to solve the tracking problem of $\Sigma_{1}$ with $\Sigma_{a_{1}}$. The formation controller of the follower can be obtain by replacing the leader $\Sigma_{0}$ with the transformed system $\Sigma_{a_{1}}$ in Theorem \ref{the_sing_tra}.

Next, the multiple followers formation tracking is taken into consideration. Similar to the case of consensus tracking, for robot $\Sigma_{i}$ ($i=1,2,\cdots,N$), we construct a virtual system $\Sigma_{c_{i}}$ in (\ref{eq_sys_c}), which is the convex combination of the parent vertexes of robot $i$. Thus, we choose the system $\Sigma_{c_{i}}$ as the virtual leader for system $\Sigma_{i}$. Let $\bar{g}_{c_{i}i}$ denote the relative configuration of $\Sigma_{i}$ with respect to its virtual leader $\Sigma_{c_{i}}$, that is \begin{equation} \bar{g}_{c_{i}i}=\begin{bmatrix} R(\bar{\theta}_{c_{i}i}) & \bar{\bm{p}}_{c_{i}i} \\ \bm{0} & 1\end{bmatrix}, \end{equation} where $\bar{\bm{p}}_{c_{i}i}=[\bar{x}_{c_{i}i}\ \ \bar{y}_{c_{i}i}]^{\trans}$ the desired relative position, and the desired relative attitude angle is decided by \begin{equation} \label{eq_bar_the_i} \bar{\theta}_{c_{i}i} = \arctan\frac{\omega_{c_{i}}\bar{x}_{c_{i}i}}{v_{x0}-\omega_{c_{i}i}\bar{y}_{c_{i}i}}. \end{equation}

For each robot $i$, the transformed system is defined by \begin{equation} \Sigma_{a_{i}}: \left\{ \begin{aligned} & \dot{g}_{a_{i}}=g_{a_{i}}\hat{\xi}_{a_{i}} \\ & \dot{\hat{\xi}}_{a_{i}}=\hat{u}_{a_{i}} \end{aligned}, \right. i=1,2,\cdots,N, \end{equation} where the transformed configuration is $g_{a_{i}}=g_{c_{i}}\bar{g}_{c_{i}i}$, the transformed velocity is $\hat{\xi}_{a_{i}} = \mbox{Ad}_{\bar{g}_{c_{i}i}^{-1}}\hat{\xi}_{c_{i}}$ and the transformed control input is $\hat{u}_{a_{i}} = \mbox{Ad}_{\bar{g}_{c_{i}i}^{-1}}\hat{u}_{c_{i}}$. Similar to the case of formation tracking of one leader and one follower, when the follower $\Sigma_{i}$ tracks the transformed system $\Sigma_{a_{i}}$, the formation specified by $\bar{g}_{c_{i}i}$ can be achieved. Therefore, the formation control law for $\Sigma_{i}$ is able to be obtained based on the tracking results. In other words, as long as the virtual leader $\Sigma_{c_{i}}$ in Theorem \ref{the_multi_tra} is replaced with the transformed system $\Sigma_{a_{i}}$, then the formation controller can be obtained. Once every mobile robot achieves the formation with its virtual leader, the goal of formation tracking of the whole networked system can be realized eventually. Therefore, we can directly propose the formation tracking strategy as follows.

\begin{theorem}[Formation Tracking]
\label{the_multi_tra_for}
Consider $N+1$ nonholonomic mobile robots with dynamics (\ref{sys}), which are connected by a directed acyclic graph with one root node. For arbitrary reference trajectory of the leader, if the control strategy is designed as \begin{equation} \label{eq_ui_for_tr_0} \begin{aligned} u_{\theta i} = &-k_{e}(\theta_{i}-\tilde{\theta}_{i}) -k_p(\theta_{a_{i}i}+k\beta_{a_{i}i}) -k_d\omega_{a_{i}i}+u_{\theta a_{i}},\\ u_{x i}=&-k_pq_{xa_{i}i}-k_dv_{xa_{i}i}+(u_{xa_{i}}-u_{\theta a_{i} }r_{ya_{i}i})\cos\theta_{a_{i}i}+u_{\theta a_{i}}r_{xa_{i}i}\sin\theta_{a_{i}i}, \end{aligned} \end{equation} in which $\tilde{\theta}_{i}=\theta_{a_{i}}+\tilde{\theta}_{a_{i}i}$, $\beta_{a_{i}i}=-\arctan(q_{ya_{i}i}/q_{xa_{i}i})$, $u_{\theta a_{i}}$ and $u_{x a_{i}}$ are the components of the transformed system control input $\bm{u}_{a_{i}}$, $i=1,2,\cdots,N$, then, all the followers are able to achieve formation tracking with the leader globally and asymptotically.
\end{theorem}

\section{Simulation Examples}

\begin{example}
\label{ex_1}
To verify the result of Theorem \ref{the_sing_tra}, we firstly consider the tracking control problem for one leader and one follower. The initial states of two mobile robots are given in Table \ref{tab_initial_1}. The reference control input of the leader is predefined as $u_{\theta0}=0.15\cos(0.4t)$ and $u_{x0}=10$, which represents the control force in rotation and translation respectively. The controller (\ref{eq_u1_v3}) is employed to make the follower achieve trajectory tracking. The simulation time is set as $T=20\mbox{s}$, and the simulation result is portrayed in Figure \ref{fig_tra_sing}, which shows the trajectories with orientation of two robots at different instants. As we can see, the follower is able to track the leader under the designed controller.
\end{example}

\begin{table}\small
\caption{Initial conditions of single follower tracking}
\centering
\begin{tabular}{lccccc}
\hline\noalign{\smallskip}
             & $\theta(0)$   & $x(0)$  & $y(0)$  & $\omega(0)$  & $v_{x}(0)$   \\
\noalign{\smallskip}\hline\noalign{\smallskip}
  leader     &  $0$         &  $0$     &   $0$     &   $0$       &   $0$         \\
  follower   &  $-\pi/2$    &  $500$   &  $-500$   &   $2$       &   $10$        \\
\noalign{\smallskip}\hline
\end{tabular}
\label{tab_initial_1}
\end{table}

\begin{figure}
\centering
\subfigure[$t=15\%T$]{
{\includegraphics[width=0.35\textwidth,trim=50 0 50 20,clip]{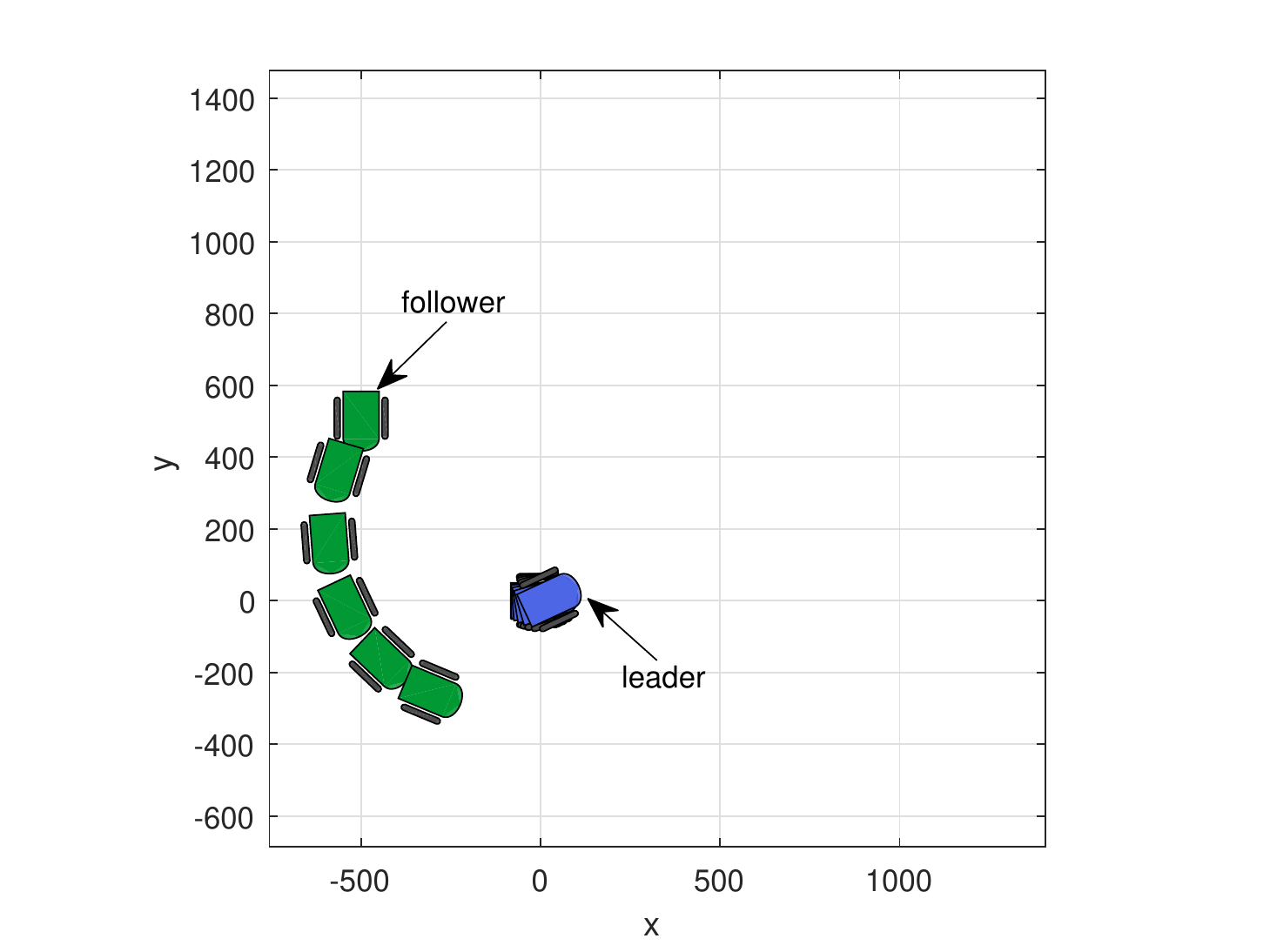}}}\qquad
\subfigure[$t=25\%T$]{
{\includegraphics[width=0.35\textwidth,trim=50 0 50 20,clip]{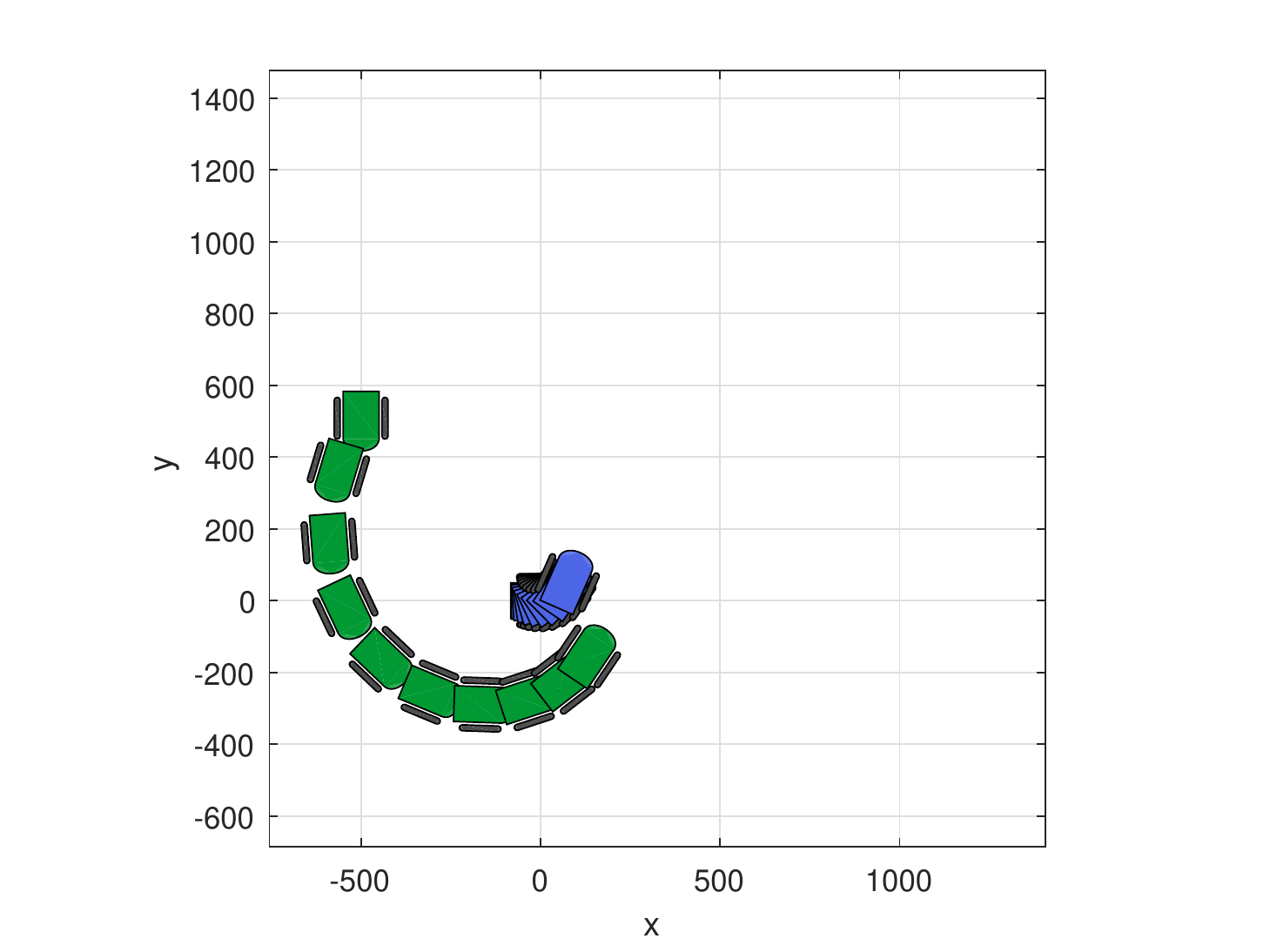}}}

\subfigure[$t=50\%T$]{
{\includegraphics[width=0.35\textwidth,trim=50 0 50 20,clip]{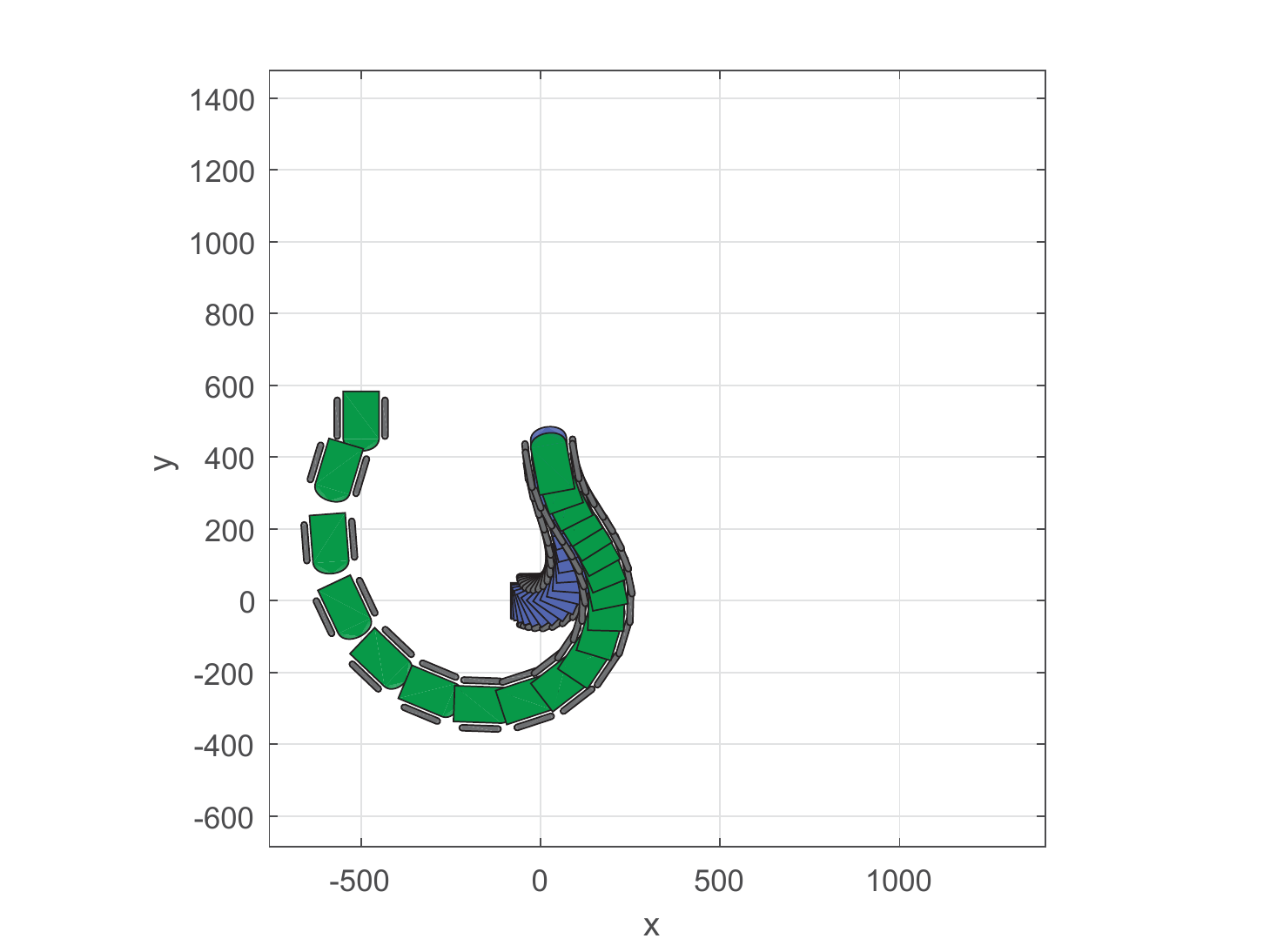}}}\qquad
\subfigure[$t=100\%T$]{
{\includegraphics[width=0.35\textwidth,trim=50 0 50 20,clip]{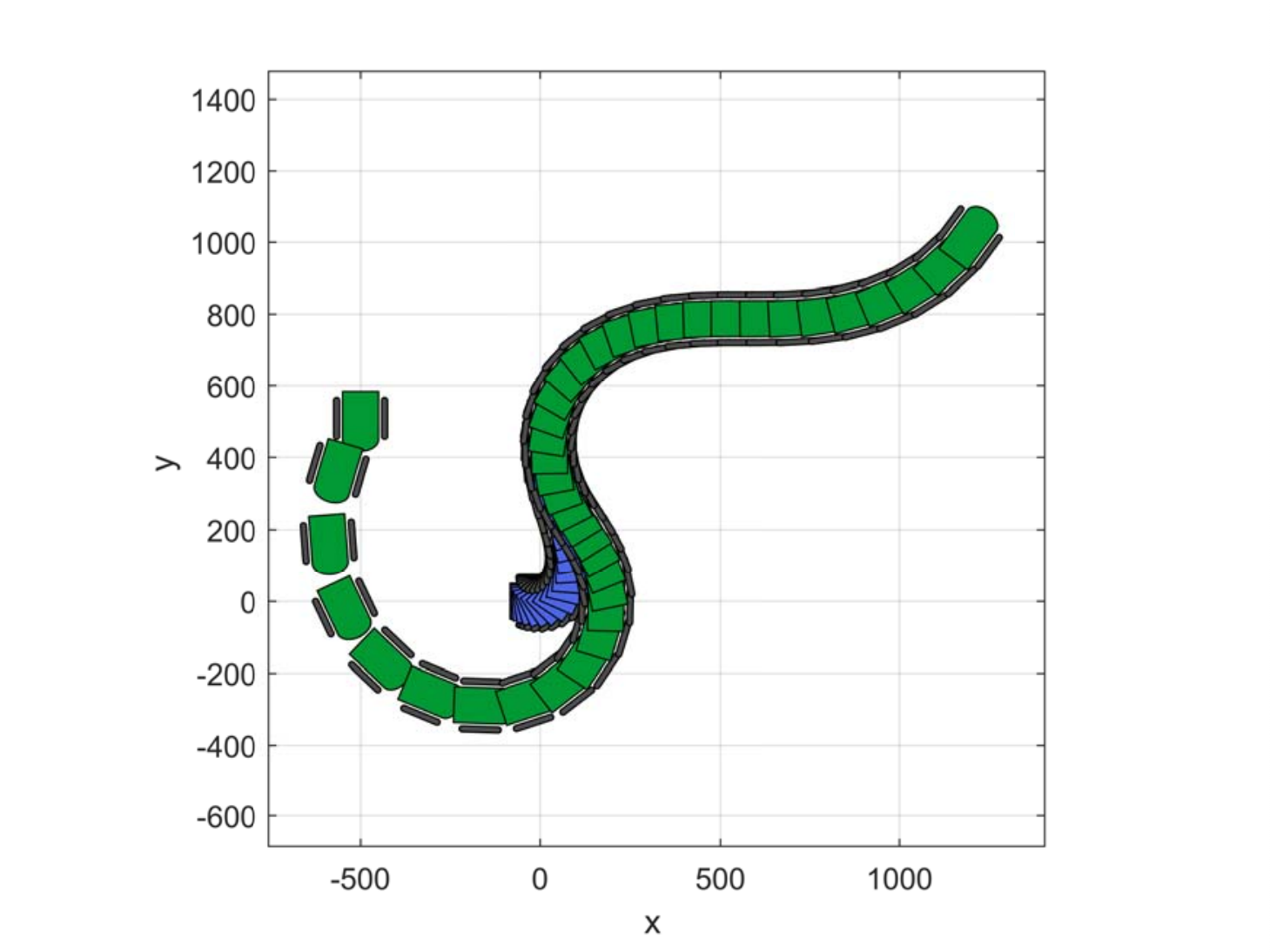}}}
\caption{Trajectories of mobile robots at different instants (smooth trajectory tracking)}
\label{fig_tra_sing}
\end{figure}

\begin{example}
In order to validate the effectiveness of the tracking controller in the presence of zero reference velocities, where the PE condition does not hold, we generate a nonsmooth trajectory with the reference control input in Figure \ref{fig_refer_u}, and Figure \ref{fig_refer_v} shows the leader's velocity under the reference control input. From the evolution of the reference velocity, there exist zero velocities ($\omega_{0}=0$ and $v_{x0}=0$) in the movement of the leader, indicating that the leader does not satisfy the PE condition. Figure \ref{fig_refer_tra} shows the reference trajectory of the leader, which is evidently nonsmooth. The initial states of the follower are defined as $\theta(0)=-\pi/2$, $x(0)=-25$, $y(0)=-10$, $\omega(0)=0$ and $v_{x}(0)=0$. The controller (\ref{eq_u1_v3}) is employed again for the follower. The simulation result is shown in Figure \ref{fig_tra_sing_non}, which illustrates the trajectories of the leader and the follower, verifying that the proposed controller is also effective for trajectories with zero velocities.
\end{example}

\begin{figure}
\centering
\subfigure[Reference control input for the leader]{\label{fig_refer_u}
{\includegraphics[width=0.48\textwidth,trim=10 0 0 0,clip]{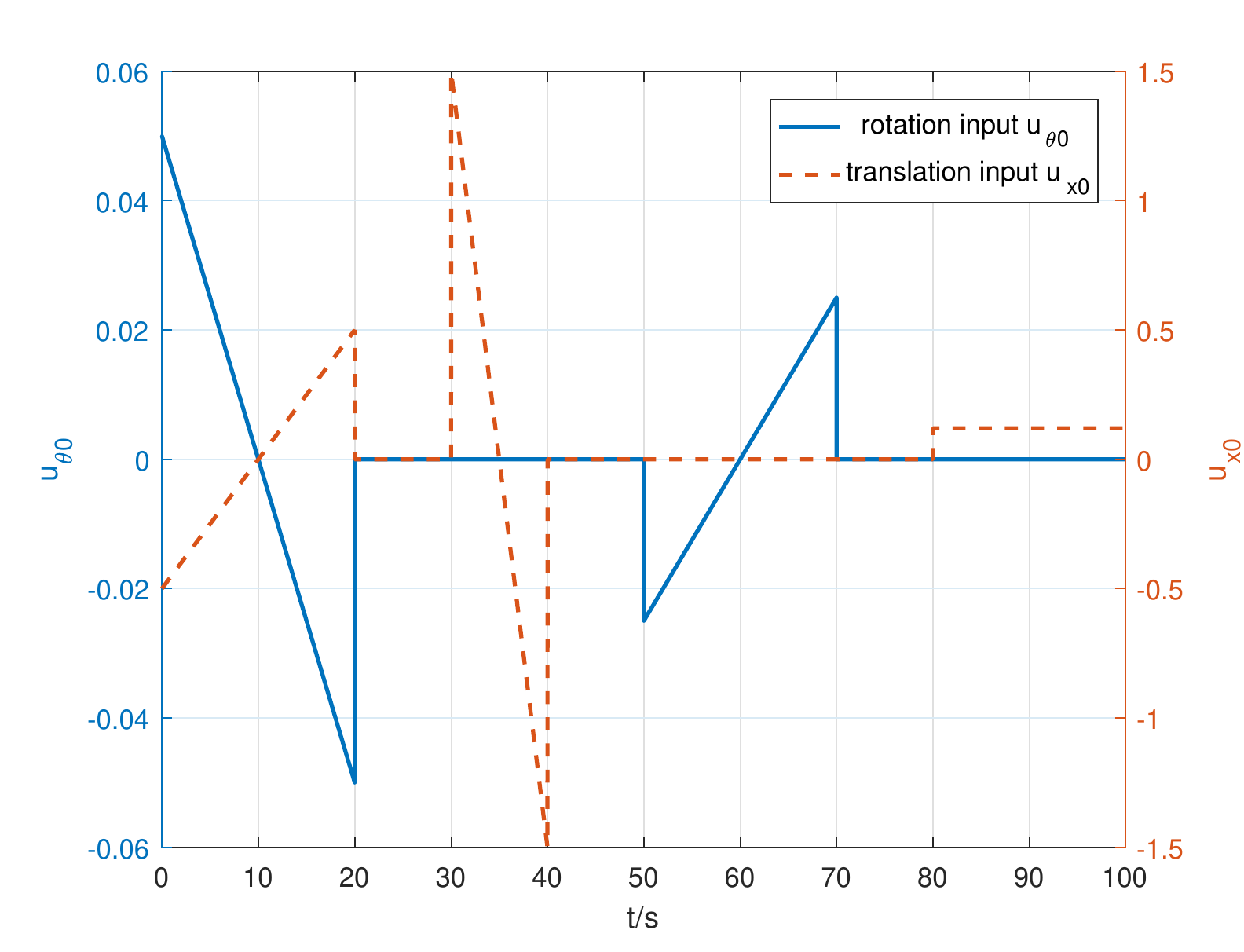}}}
\subfigure[Velocity of the leader corresponding to input]{\label{fig_refer_v}
{\includegraphics[width=0.48\textwidth,trim=10 0 0 0,clip]{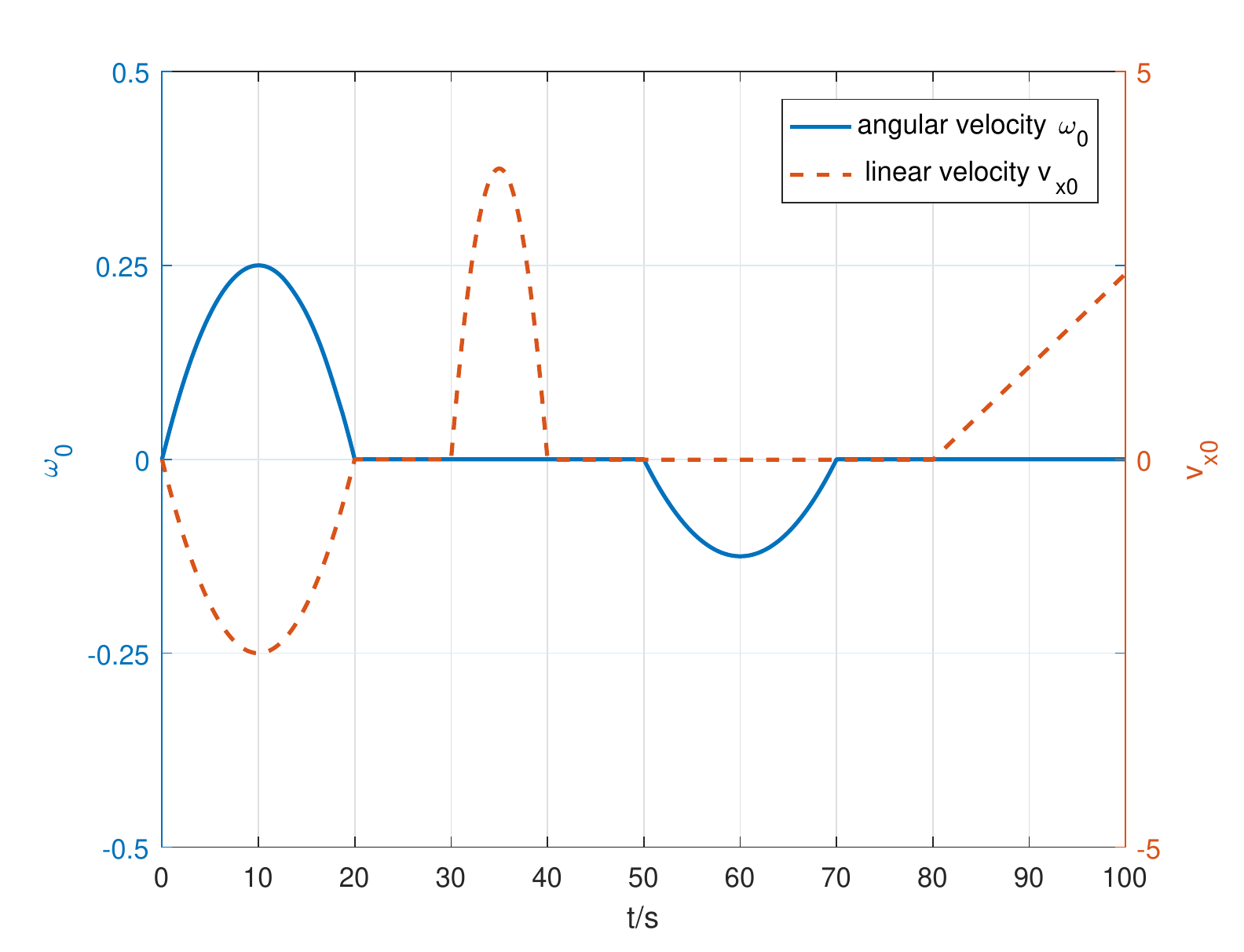}}}
\caption{Reference control input and corresponding velocity of the leader}
\end{figure}

\begin{figure}
\centering
\begin{minipage}{0.48\textwidth}
\centering
\includegraphics[width=1\textwidth,trim=40 5 20 20,clip]{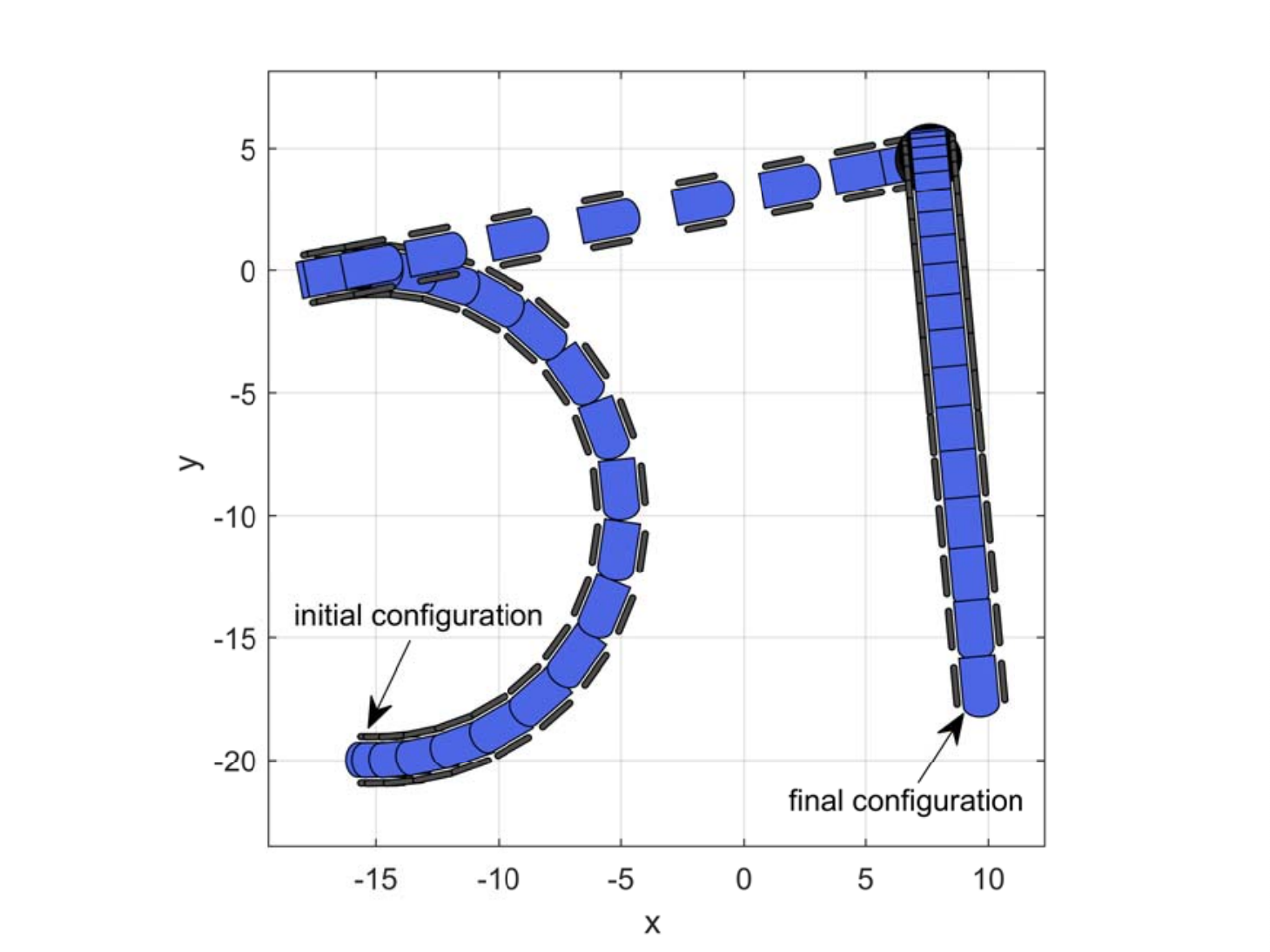}
\caption{Nonsmooth trajectory of the leader under reference inputs}
\label{fig_refer_tra}
\end{minipage}
\quad
\begin{minipage}{0.48\textwidth}
\centering
\includegraphics[width=1\textwidth,trim=40 5 20 20,clip]{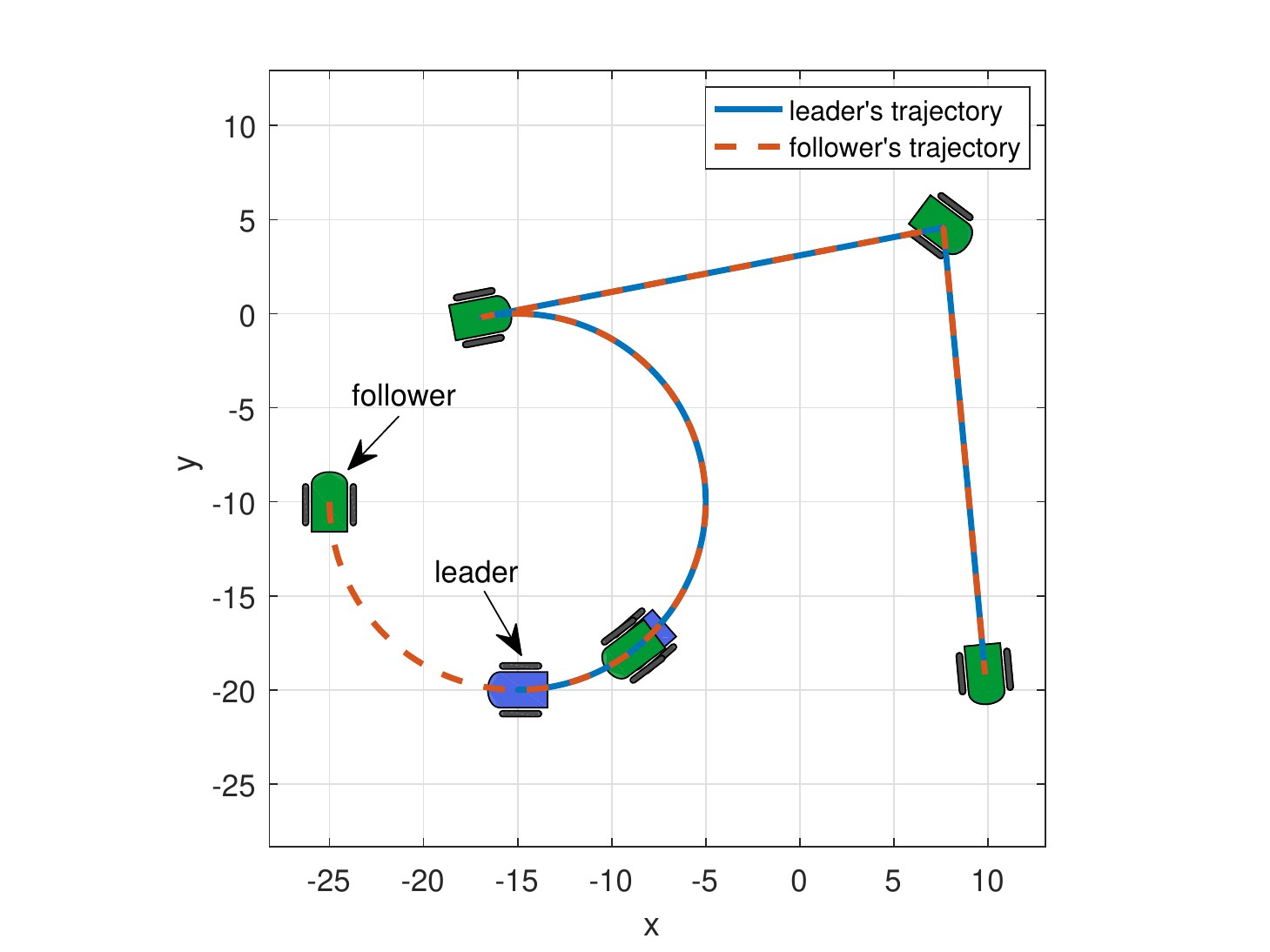}
\caption{Trajectories of mobile robots (nonsmooth trajectory tracking)}
\label{fig_tra_sing_non}
\end{minipage}
\end{figure}

\begin{example}
To show the global convergence property, the case of leader and follower with opposite heading is verified. It is shown in this example that the initial relative attitude angle specified to $\pi$ or $-\pi$ would correspond to different trajectories. We employ the controller in Theorem \ref{the_sing_tra} to make the follower track the leader. The initial configuration of the leader is same as that in Example \ref{ex_1}, and the reference input is defined as $u_{\theta0}=0$ and $u_{x0}=0.1$. The simulation time is set as $T=10\mbox{s}$. Firstly, the initial condition of the follower is chosen to be that $\theta(0)=\pi$, $x(0)=-10$, $y(0)=0$, $\omega(0)=0$, $v_{x}(0)=0$, and we employ the tracking control to acquire the simulation result. Then, we change the initial attitude angle $\theta(0)$ to $-\pi$, and apply the control law again. Finally, these two trajectories are plotted into the same figure, which is shown in Figure \ref{fig_pi}. In fact, the above two initial conditions correspond to the same one configuration, which is marked in yellow color. As is portrayed in the figure, the robot rotates in clockwise direction with $\theta(0)=\pi$, while in anticlockwise direction with $\theta(0)=-\pi$. Thus, if there exist surrounding limitations for movement, the proposed control law can provide choices for a suitable trajectory, which is a significant advantage in real applications.
\end{example}

\begin{figure}
\centering
\subfigure[$t=25\%T$]{
{\includegraphics[width=0.45\textwidth,trim=50 0 50 0,clip]{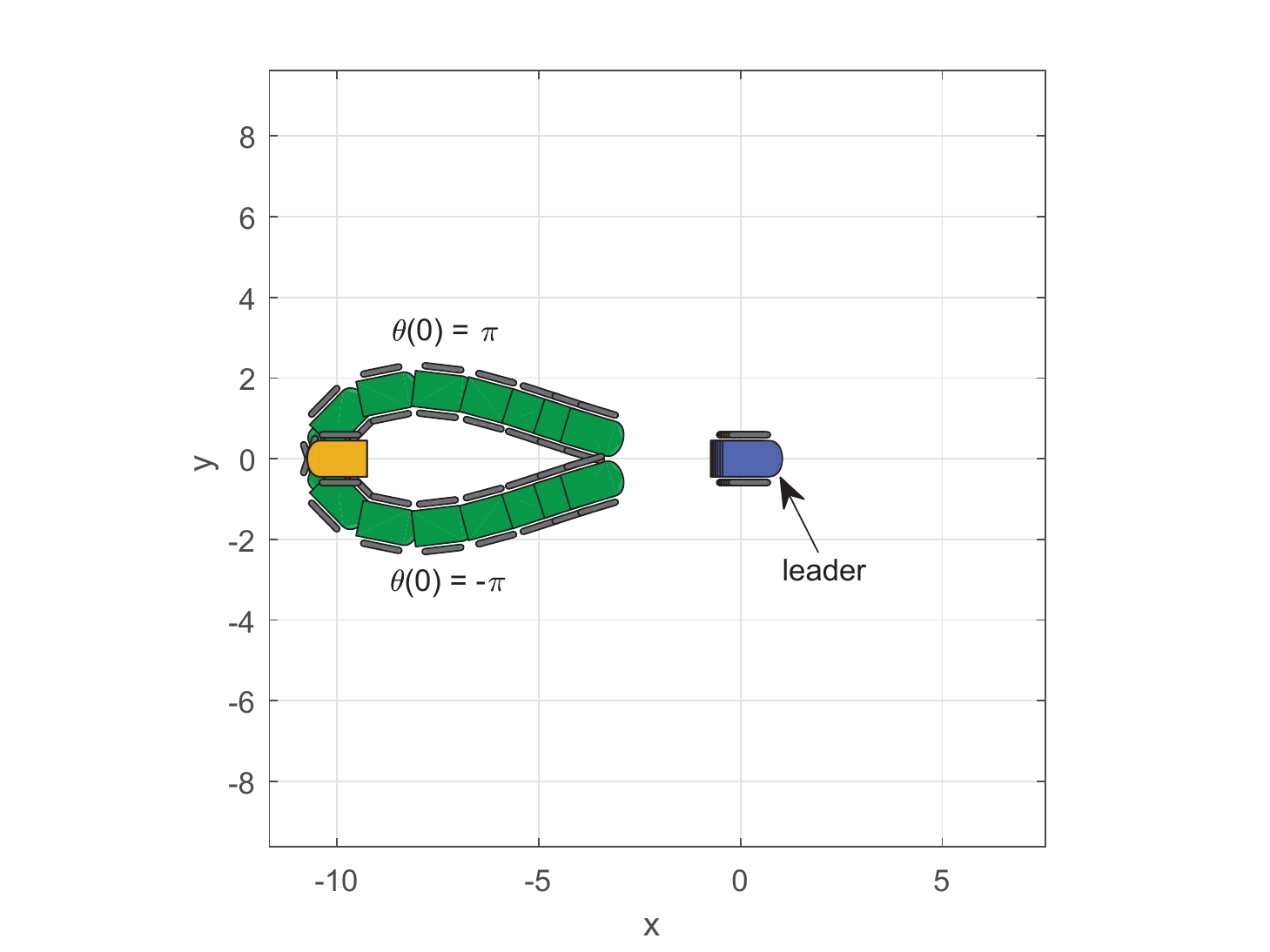}}}
\subfigure[$t=100\%T$]{
{\includegraphics[width=0.45\textwidth,trim=50 0 50 0,clip]{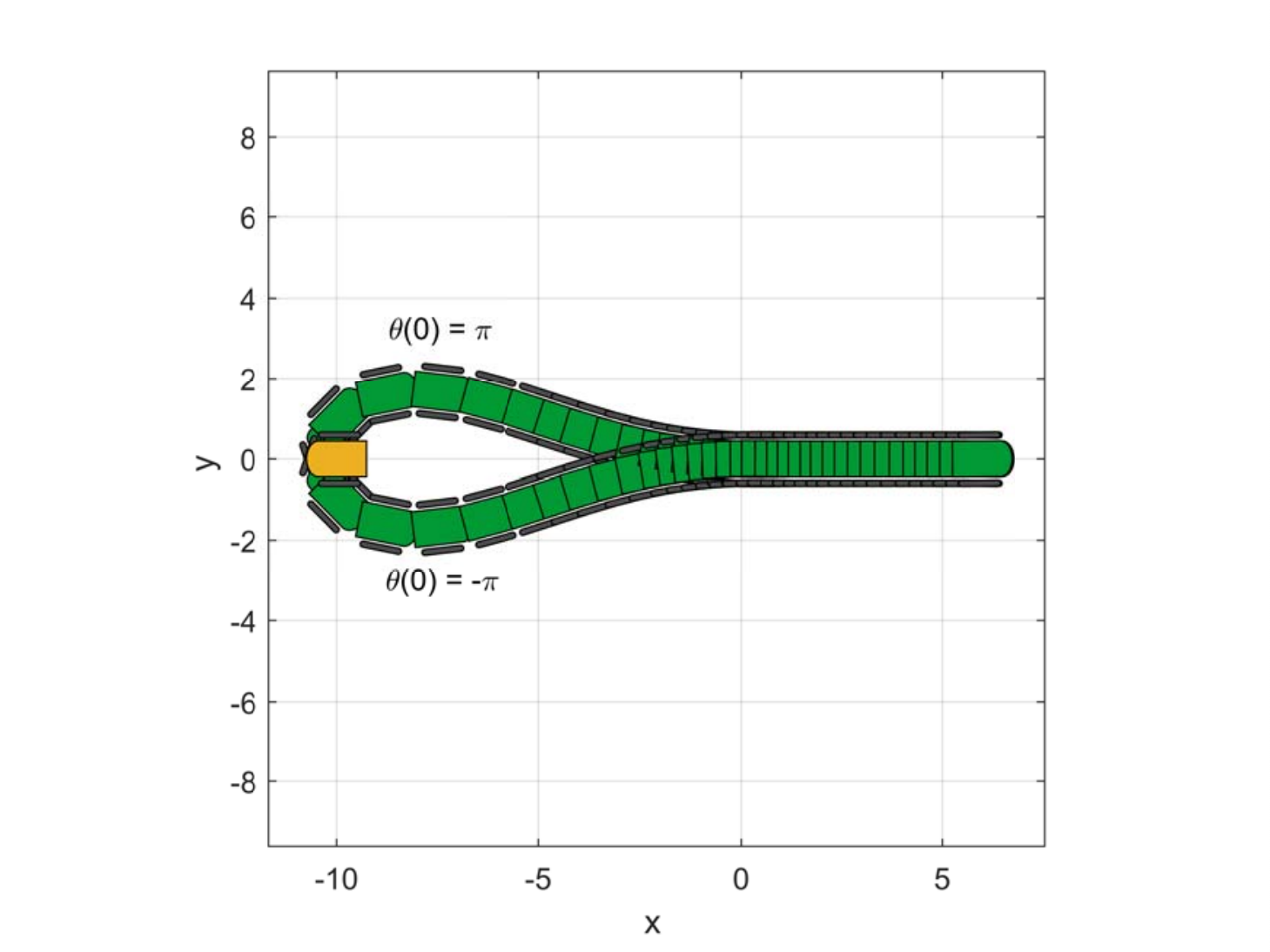}}}
\caption{Trajectories with initial attitude angle $\pi$ and $-\pi$}
\label{fig_pi}
\end{figure}

\begin{example}
This example is to verify the consensus tracking controller (\ref{eq_ui_con_tr_0}). We consider 4 nonholonomic mobile robots, in which the leader is labeled as 0 and the followers are numbered as 1, 2 and 3. The communication topology among the robots is depicted in Figure \ref{fig_topol}, which shows that the robot No.3 has two local leaders (or parent nodes). The reference control input of the leader is predefined as $u_{\theta0}=0.1\sin(0.4t)$ and $u_{x0}=1$. The initial states of all the robots are given in Table \ref{tab_initial_2}. We set the simulation time as $T=10\mbox{s}$. The robots' trajectories with orientation at different instants is demonstrated in Figure \ref{fig_tra_con}. It is can be seen that all the followers are able to achieve consensus with the leader, in other words, the consensus tracking is obtained successfully.
\end{example}

\begin{figure}
\centering
\includegraphics[width=0.15\textwidth]{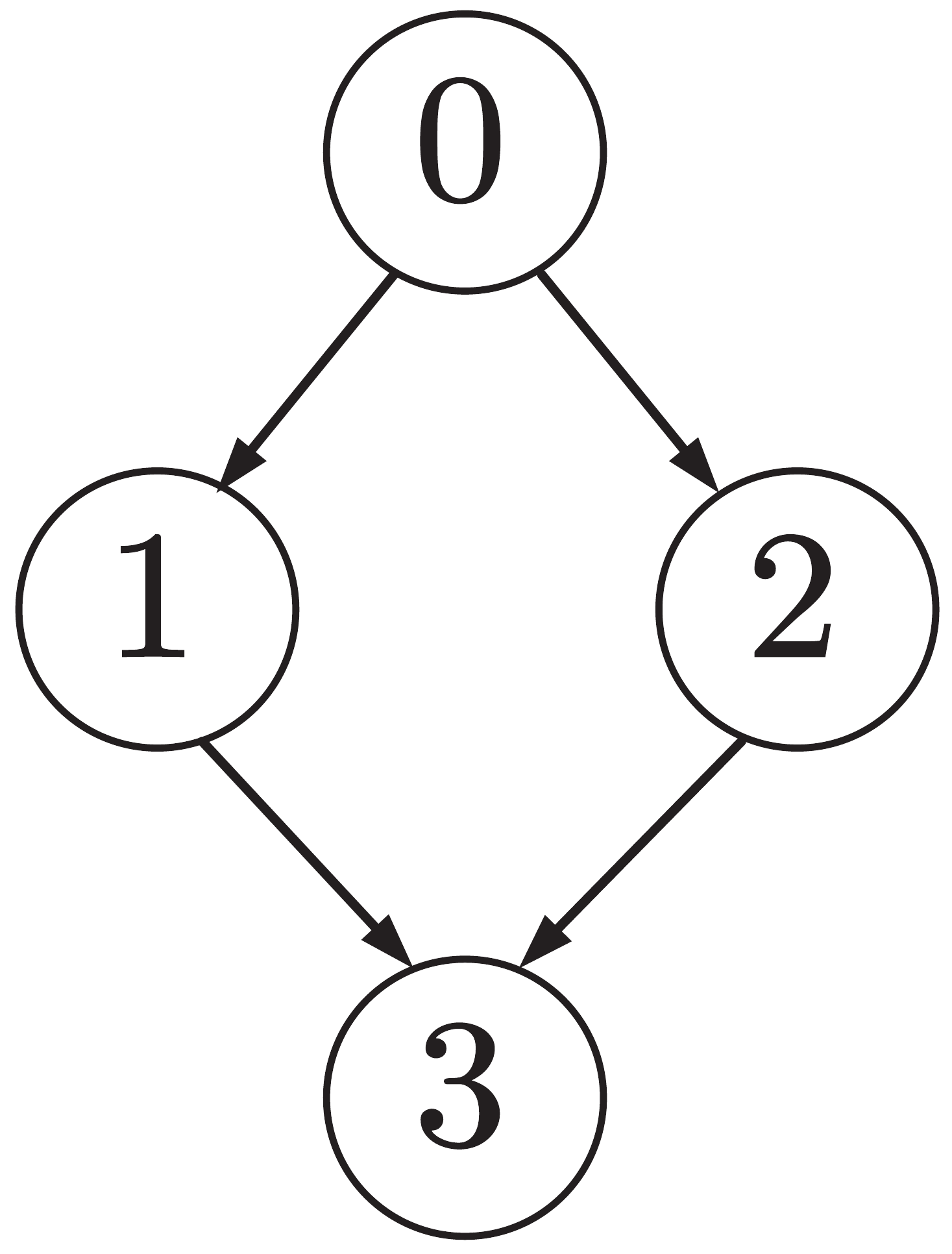}
\caption{The communication topology of mobile robots}
\label{fig_topol}
\end{figure}

\begin{table}\small
\caption{Initial conditions of consensus tracking}
\centering
\begin{tabular}{lccccc}
\hline\noalign{\smallskip}
  No. & $\theta(0)$   & $x(0)$  & $y(0)$  & $\omega(0)$  & $v_{x}(0)$   \\
\noalign{\smallskip}\hline\noalign{\smallskip}
  0   &  $0$         &  $0$     &   $0$      &   $0$       &   $0$         \\
  1   &  $-\pi/4$    &  $-10$   &  $10$      &   $0$       &   $0$         \\
  2   &  $\pi/4 $    &  $-15$   &  $20$      &   $0$       &   $0$         \\
  3   &  $\pi$       &  $-30$   &  $-5$      &   $0$       &   $0$         \\
\noalign{\smallskip}\hline
\end{tabular}
\label{tab_initial_2}
\end{table}

\begin{figure}
\centering
\subfigure[$t=15\%T$]{
{\includegraphics[width=0.35\textwidth,trim=50 0 50 20,clip]{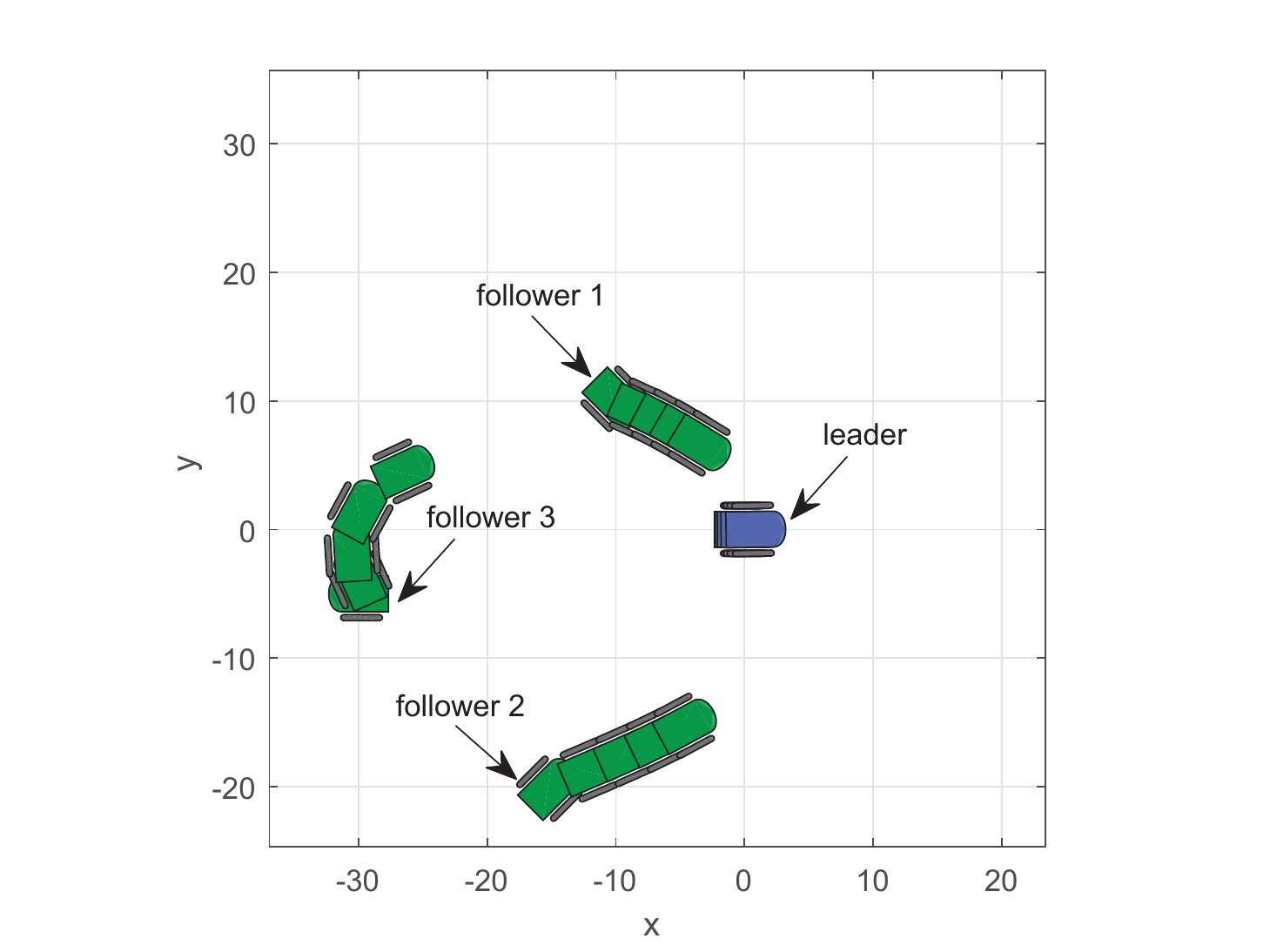}}}\qquad
\subfigure[$t=25\%T$]{
{\includegraphics[width=0.35\textwidth,trim=50 0 50 20,clip]{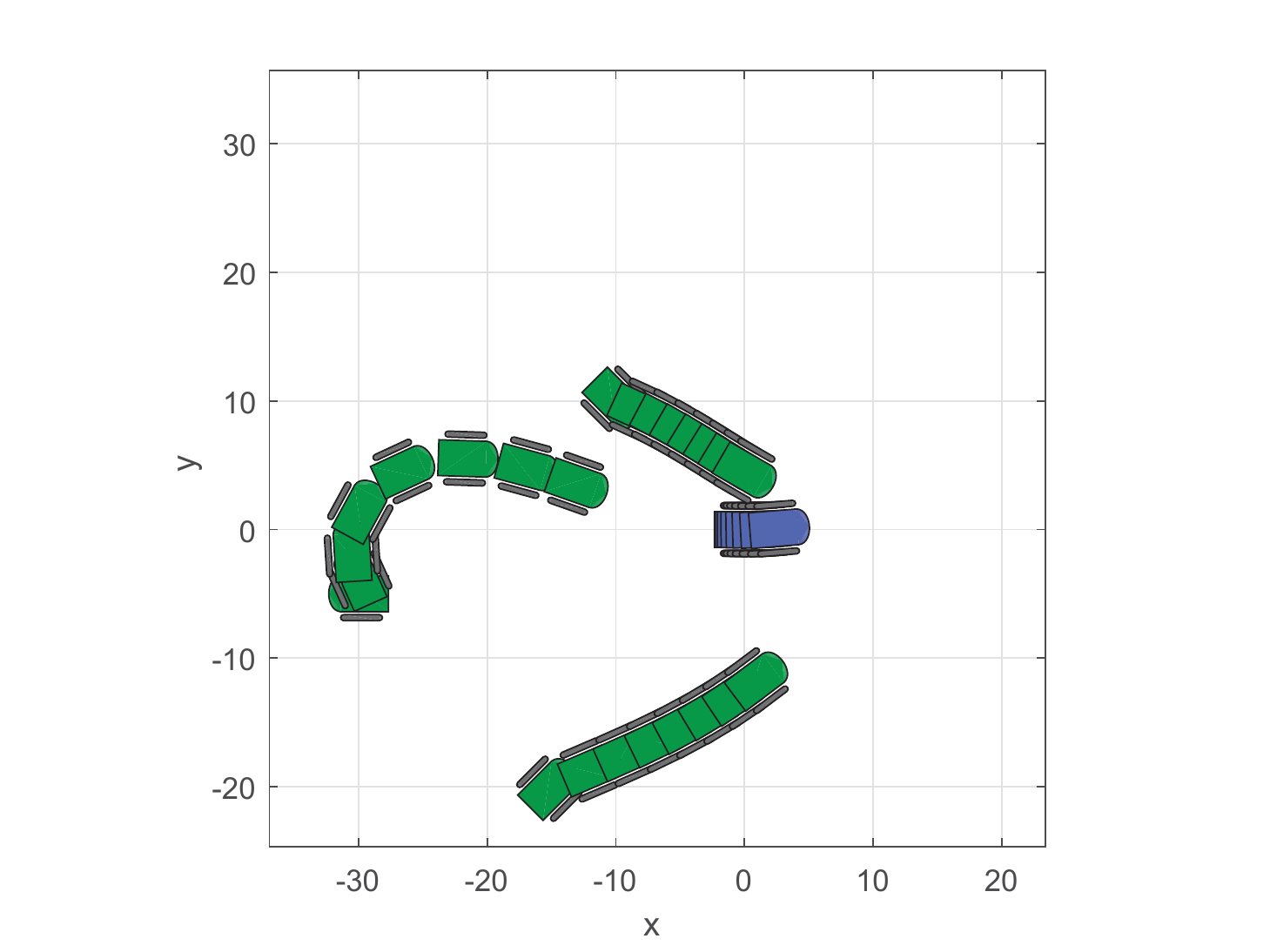}}}

\subfigure[$t=50\%T$]{
{\includegraphics[width=0.35\textwidth,trim=50 0 50 20,clip]{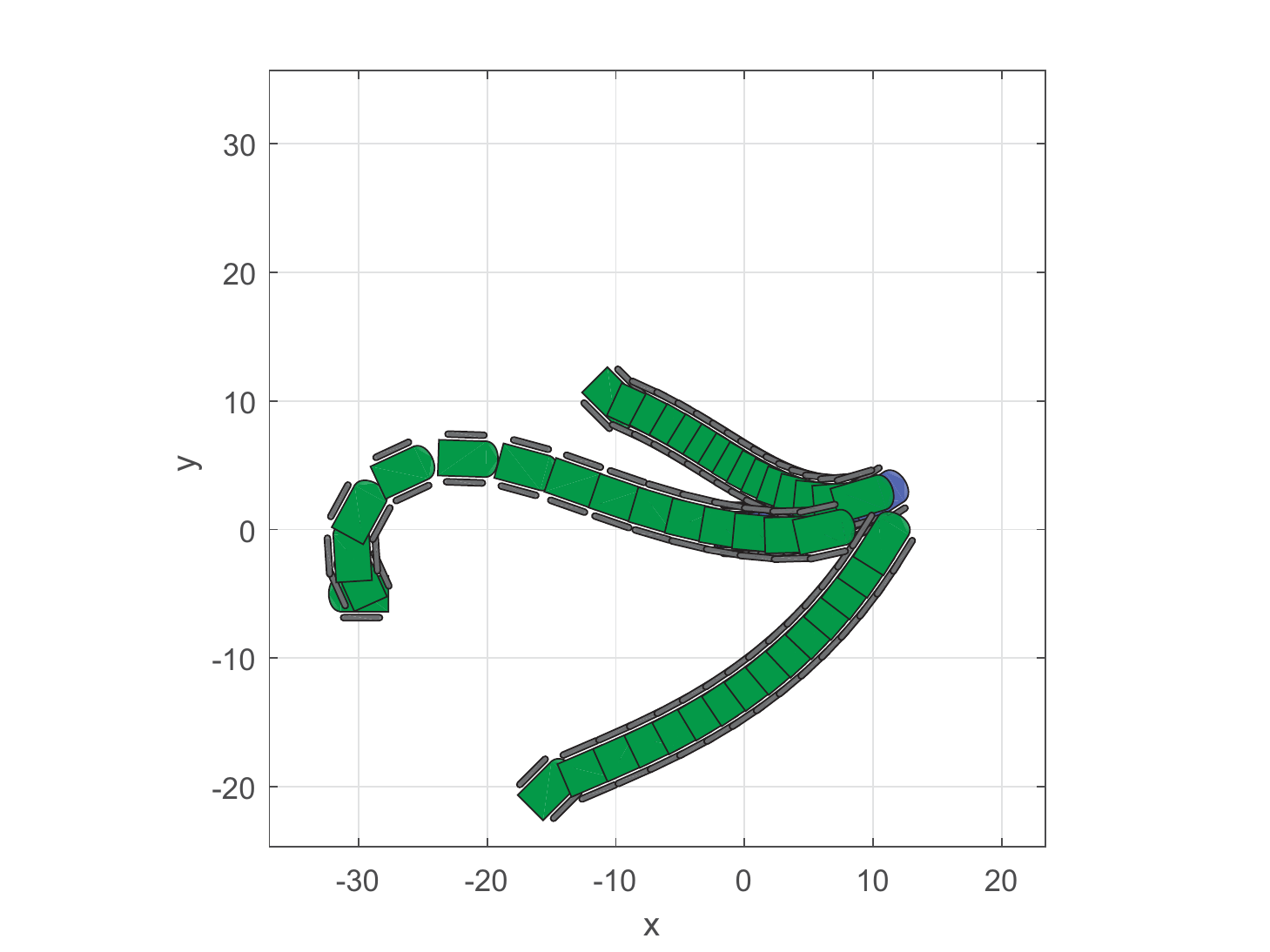}}}\qquad
\subfigure[$t=100\%T$]{
{\includegraphics[width=0.35\textwidth,trim=50 0 50 20,clip]{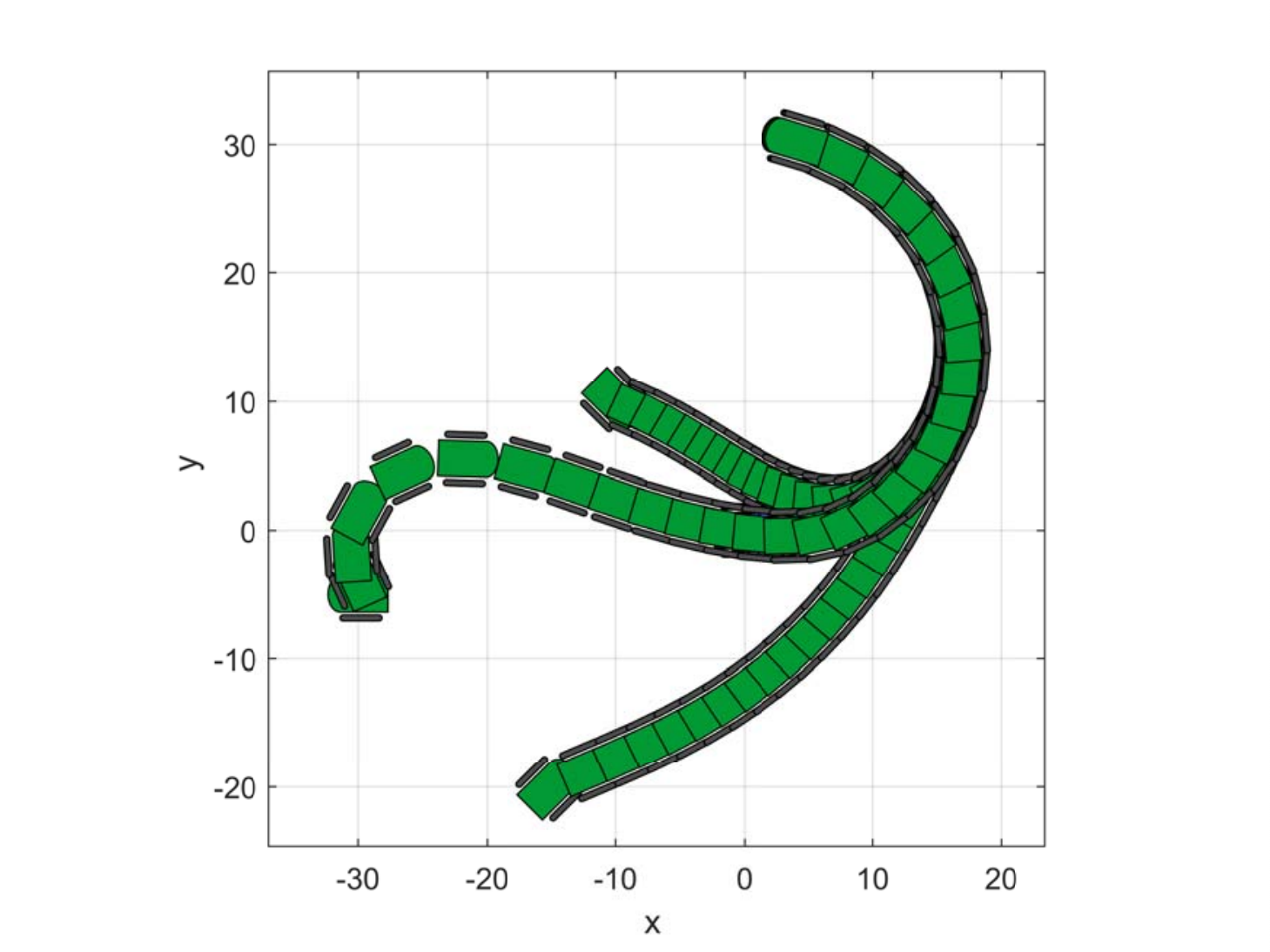}}}
\caption{Trajectories of mobile robots at different instants (consensus tracking)}
\label{fig_tra_con}
\end{figure}

\begin{example}
Finally, we provide an example to illustrate the effectiveness of the formation tracking controller (\ref{eq_ui_for_tr_0}). The number of the nonholonomic mobile robots is 4 and the communication topology is still Figure \ref{fig_topol}. The initial states are given in Table \ref{tab_initial_3}, and the reference control input of the leader is predetermined as \begin{equation*} u_{\theta0}=\left\{ \begin{aligned} 0 \qquad\quad 0\leq t \leq 3 \\ 0.1\sin(0.5t-1.5) \quad t\geq 3 \end{aligned} \right. , \quad u_{x0}=1. \end{equation*} We define the desired formation as follows \begin{equation*} \bar{p}_{c_{1}1}=[-15\ \ 15]^{\trans}, \quad \bar{p}_{c_{2}2}=[-15\ \ -15]^{\trans},\quad \bar{p}_{c_{3}3}=[-15\ \ 0]^{\trans}, \end{equation*} which represents the desired position of the follower with respect to its (virtual) local leader. The simulation time is set as $T=10\mbox{s}$, and we portray the trajectories of the mobile robots in Figure \ref{fig_tra_form}. It is shown in the pictures that all the followers achieve the desired formation shape with the leader.
\end{example}

\begin{table}\small
\caption{Initial conditions of formation tracking}
\centering
\begin{tabular}{lccccc}
\hline\noalign{\smallskip}
  No. & $\theta(0)$   & $x(0)$  & $y(0)$  & $\omega(0)$  & $v_{x}(0)$   \\
\noalign{\smallskip}\hline\noalign{\smallskip}
  0   &  $0$         &  $0$     &   $0$      &   $0$       &   $0$         \\
  1   &  $-\pi/2$    &  $-40$   &  $40$      &   $0$       &   $0$         \\
  2   &  $\pi/2 $    &  $-20$   &  $-25$     &   $0$       &   $0$         \\
  3   &  $\pi$       &  $-70$   &  $-10$     &   $0$       &   $0$         \\
\noalign{\smallskip}\hline
\end{tabular}
\label{tab_initial_3}
\end{table}

\begin{figure}
\centering
\subfigure[$t=25\%T$]{
{\includegraphics[width=0.35\textwidth,trim=50 0 50 20,clip]{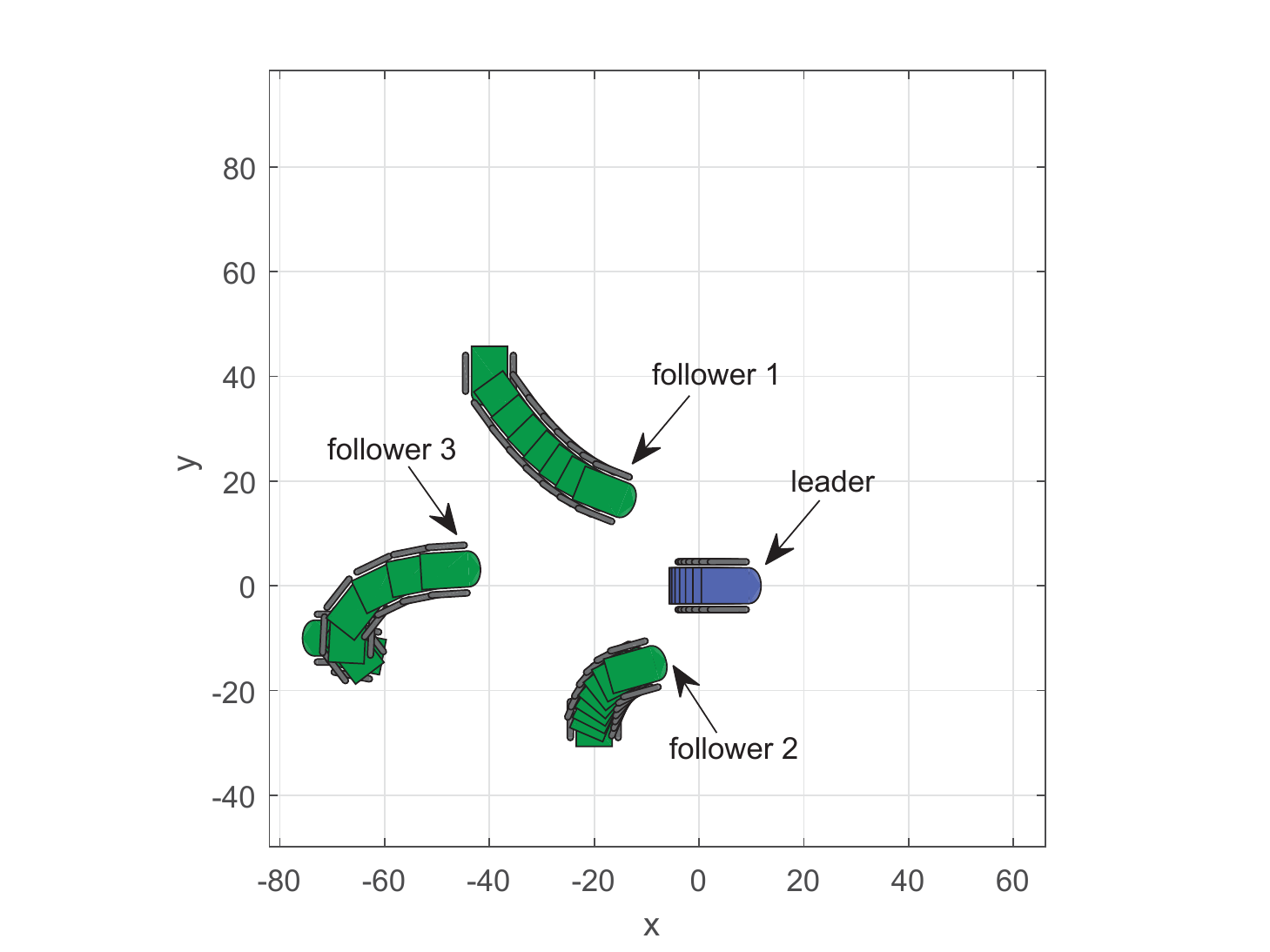}}}\qquad
\subfigure[$t=50\%T$]{
{\includegraphics[width=0.35\textwidth,trim=50 0 50 20,clip]{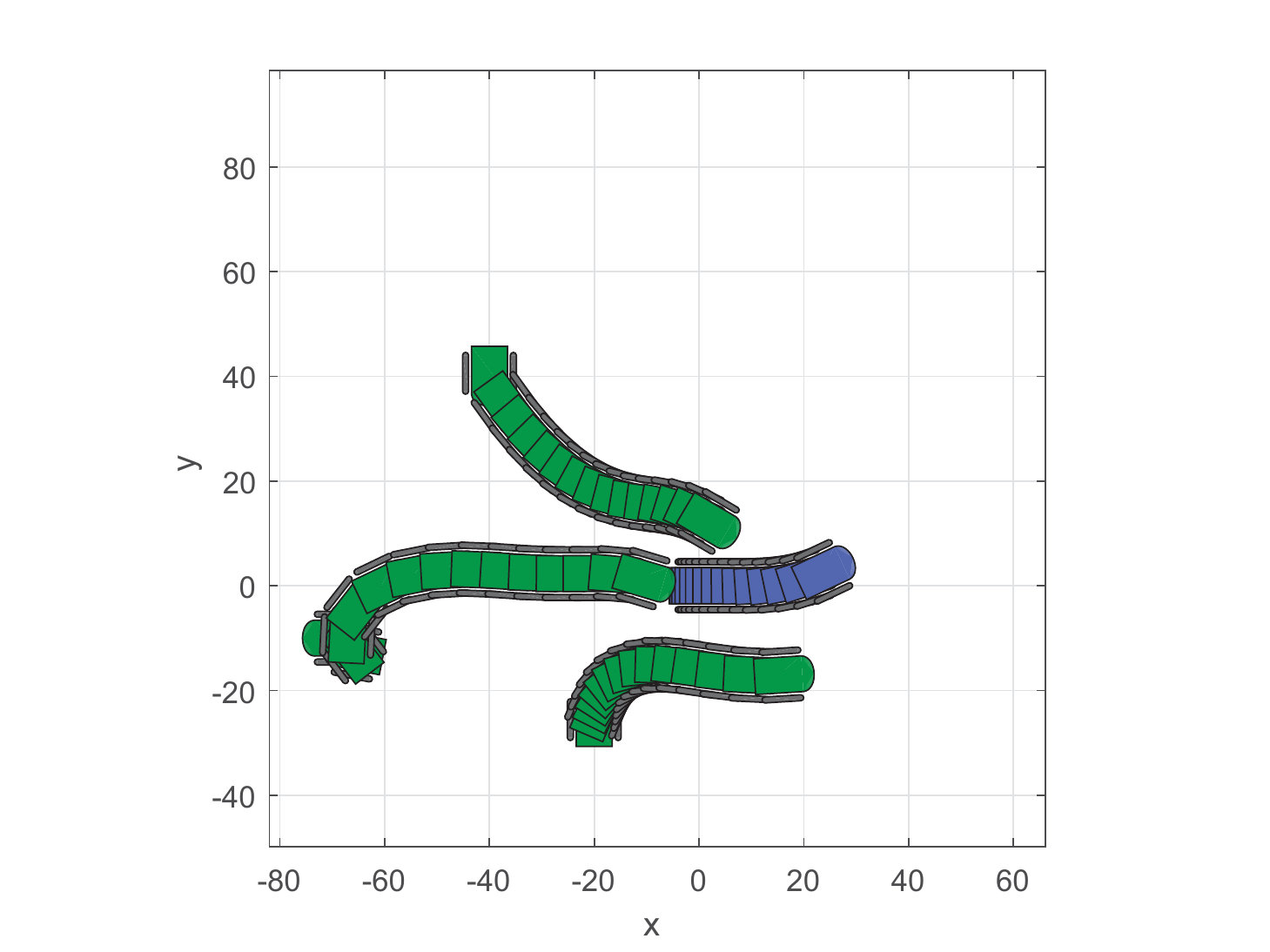}}}

\subfigure[$t=75\%T$]{
{\includegraphics[width=0.35\textwidth,trim=50 0 50 20,clip]{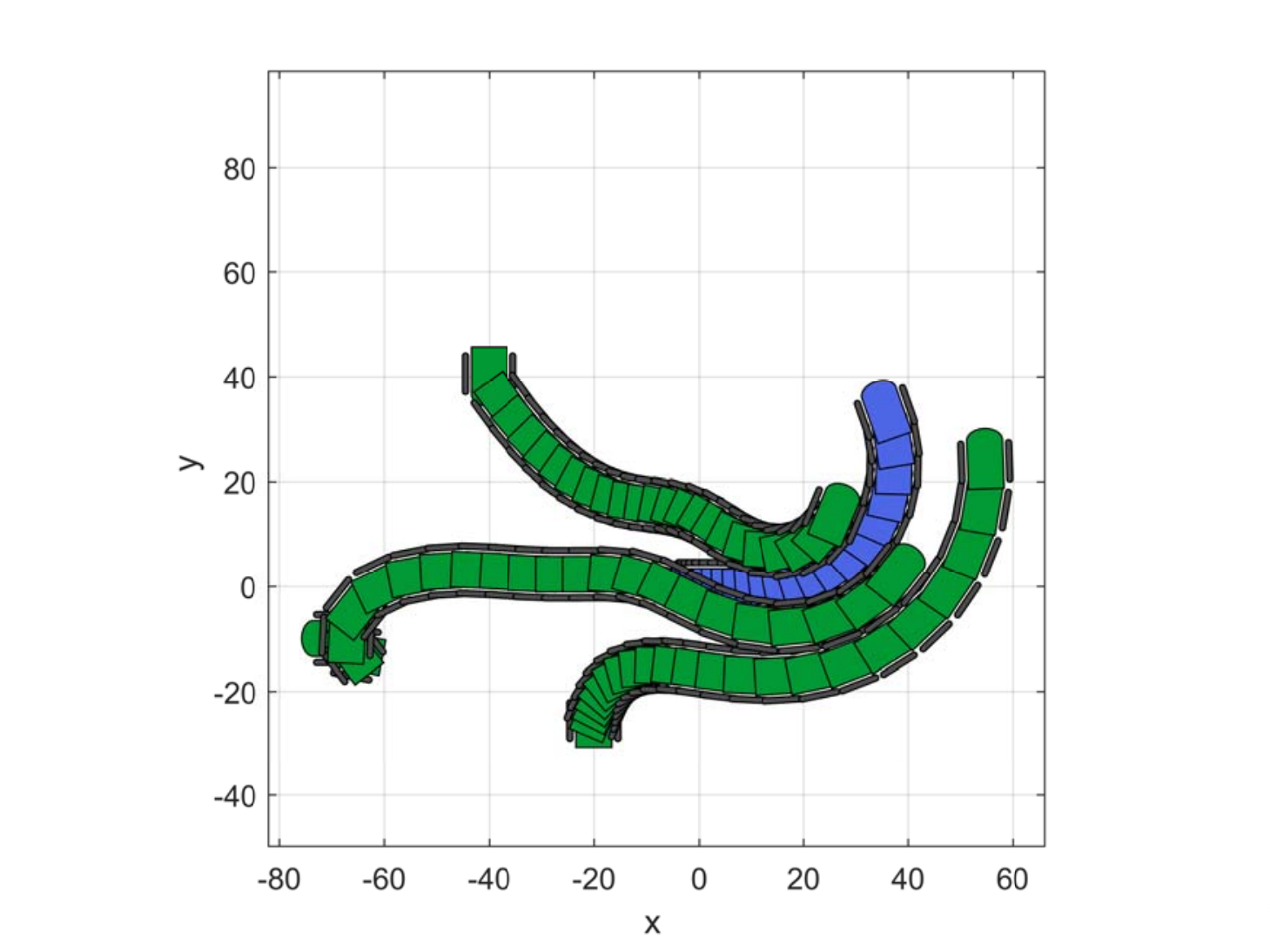}}}\qquad
\subfigure[$t=100\%T$]{
{\includegraphics[width=0.35\textwidth,trim=50 0 50 20,clip]{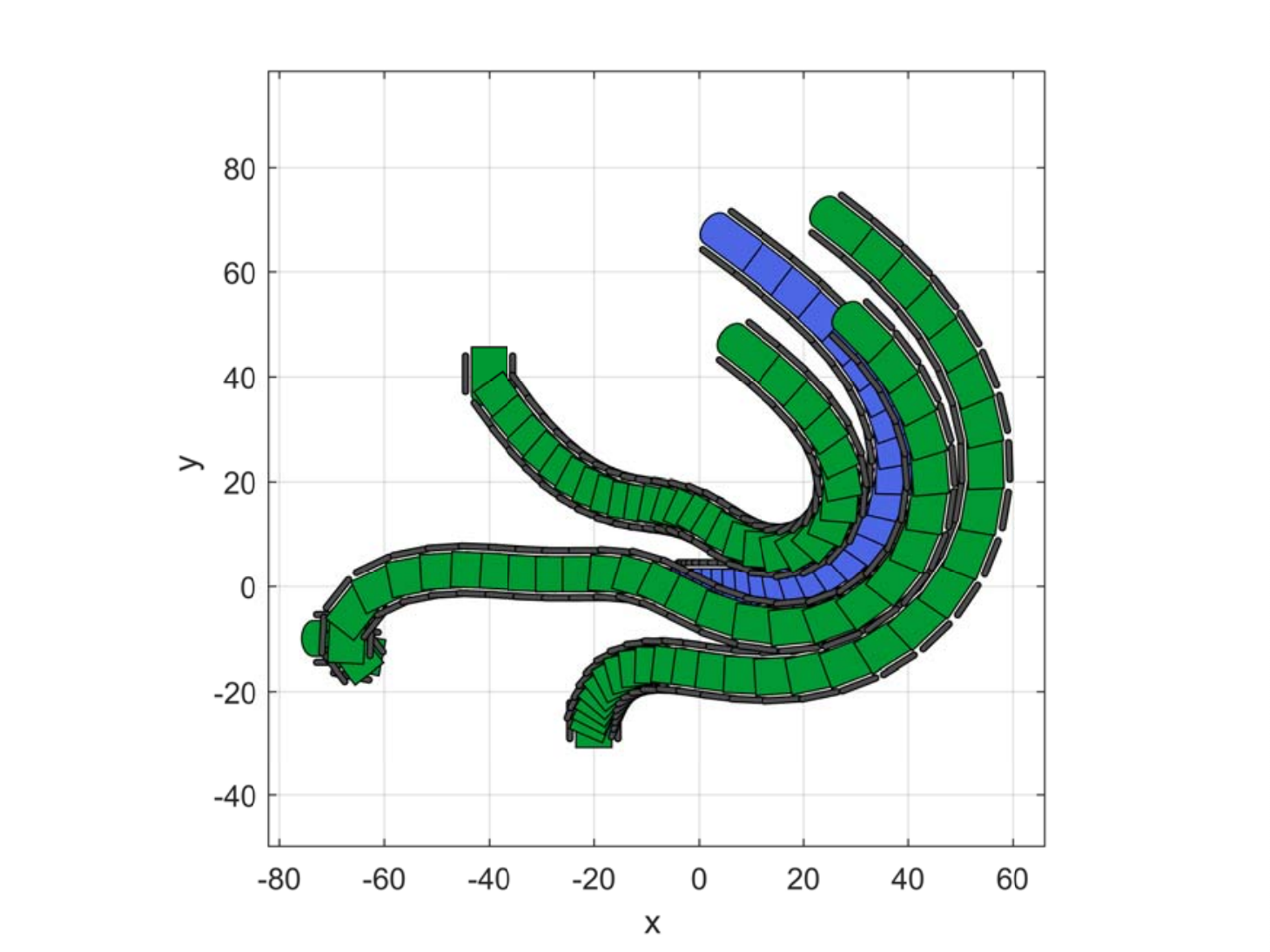}}}
\caption{Trajectories of mobile robots at different instants (formation tracking)}
\label{fig_tra_form}
\end{figure}

\section{Conclusion}

In this paper, we consider the trajectory tracking control problem for nonholonomic mobile robots based on second order dynamics. At first, the single follower tracking controller is presented by the stabilization of two relative subsystems. Later, the control strategy is extended to the consensus tracking and formation tracking problem of multiple robots, which are connected by a directed acyclic communication graph. The greatest novelty of this paper is that the reference trajectory can be arbitrarily chosen, in the sense that the condition of persistency excitation and any other requirements are not imposed on the leader. Besides, the proposed controller possesses the global convergence property.

Based on the obtained results, it is of interest to further investigate the networked system of nonholohomic mobile robots under more practical situation. Firstly, the collision between the robots should be avoided absolutely in practice. Moreover, the robots generally move in the obstacle circumstance. Thus, it is worth studying the problems of collision avoidance and obstacle avoidance in the following. Secondly, the communication topology in this paper is relatively simple. In reality, the robots might communicate with each other in a more complicated network. Therefore, the cooperative control under more general communication topology will be considered in the future.

\section*{Acknowledgements}

This work is supported by the National Natural Science Foundation of China under Grant 61773024.

\end{document}